\documentclass[12pt, draftclsnofoot, onecolumn]{IEEEtran}
\usepackage{amsmath,amsfonts}
\usepackage[lined,boxed,commentsnumbered, ruled]{algorithm2e}
\usepackage{amssymb}
\usepackage{amsthm}
\usepackage{amsmath}
\usepackage{array}
\usepackage[caption=false,font=normalsize,labelfont=sf,textfont=sf]{subfig}
\usepackage{textcomp}
\usepackage{stfloats}
\usepackage{url}
\usepackage{verbatim}
\usepackage{graphicx}
\usepackage{cite}
\usepackage{bm}
\usepackage{xcolor}

\definecolor{color1}{HTML}{D0B22B}
\definecolor{dred}{RGB}{128,0,0}
\definecolor{colorhkust}{RGB}{20,43,140}
\definecolor{colorshanghaitech}{RGB}{162,0,5}
\definecolor{colortsinghua}{RGB}{116,52,129}
\definecolor{colordark}{RGB}{184,134,11}
\usepackage{amsthm}
\theoremstyle{definition}
\newtheorem{lemma}{Lemma}
\newtheorem{theorem}{Theorem}

\newtheorem{remark}{Remark}
\newtheorem{assumption}{Assumption}

\newcommand{\transpose}{\mathsf{T}}
\newcommand{\Htranspose}{\mathsf{H}}
\newcommand{\norm}[1]{\left\|{#1}\right\|}

\newcommand{\expp}[1]{\mathbb{E}\left[{#1}\right]}

\SetKwProg{Forpara}{for each }{ do in parallel}{end}
\SetKw{Initialization}{Initialization:}

\begin{document}

%\title{Towards Green Federated Learning Over Cloud Radio Access Network}
\title{Green Federated Learning Over Cloud-RAN with Limited Fronthual Capacity and Quantized Neural Networks}

\author{Jiali Wang, \textit{Student Member, IEEE},
Yijie Mao, \textit{Member, IEEE}, \\Ting Wang, \textit{Senior Member, IEEE}, and Yuanming~Shi, \textit{Senior Member, IEEE}
\thanks{J. Wang and T. Wang are with the MoE Engineering Research Center of Software/Hardware Co-design Technology and Application; the Shanghai Key Laboratory of Trustworthy Computing;  Software Engineering Institute, East China Normal University, Shanghai 200062, China. (email: 51215902015@stu.ecnu.edu.cn, twang@sei.ecnu.edu.cn)}
\thanks{Y. Mao and Y. Shi are with the School of Information Science and Technology, ShanghaiTech University, Shanghai
201210, China. (e-mail: \{maoyj, shiym\}@shanghaitech.edu.cn)}
\\
\thanks{(\textit{Corresponding authors: Ting Wang, Yijie Mao})}
}

\maketitle

\begin{abstract}

In this paper, we propose an energy-efficient federated learning (FL) framework for the energy-constrained devices over cloud radio access network (Cloud-RAN), where each device adopts quantized neural networks (QNNs) to train a local FL model and transmits the quantized model parameter to the remote radio heads (RRHs). Each RRH receives the signals from devices over the wireless link and forwards the signals to the CS via the fronthaul link. We rigorously develop an energy consumption model for the local training at devices through the use of QNNs and communication models over Cloud-RAN. Based on the proposed energy consumption model, we formulate an energy minimization problem that optimizes the fronthaul rate allocation, user transmit power allocation, and QNN precision levels while satisfying the limited fronthaul capacity constraint and ensuring the convergence of the proposed FL model to a target accuracy.
%with respect to fronthaul rate allocation, power allocation and precision levels of QNNs while ensuring convergence subject to a target accuracy requirement and limited fronthaul capacity. 
To solve this problem, we analyze the convergence rate and propose efficient algorithms based on the alternative optimization technique. Simulation results show that the proposed FL framework can significantly reduce energy consumption compared to other conventional approaches. We draw the conclusion that the proposed framework holds great potential for achieving a sustainable and environmentally-friendly FL in Cloud-RAN.
\end{abstract}


\section{Introduction}
Federated learning (FL) is a distributed learning paradigm that enables wireless devices to cooperatively train a centralized model on a central CS (CS) while utilizing their respective local datasets \cite{10.1145/3298981,9606720,DBLP:journals/cm/NiknamDR20,DBLP:journals/corr/KonecnyMR15,shi2023taskoriented}.
%common and widely adopted form of FL is to learn the global model by 
The most popular approach in FL involves the learning of a global model collaboratively through iterating the following three steps: (i) the CS broadcasts the global model to the participated wireless devices; (ii) each device utilizes its local dataset to update its local model and then transmits its model updates to the CS; (iii) the CS aggregates the model updates sent by the devices, and updates the global model parameters for the next iteration of FL. FL based on the aforementioned learning steps preserves the privacy of the wireless devices, since the CS lacks direct access to their local datasets. Thanks to this unique feature, FL is rendered as an ideal learning framework in various practical applications.
%, including e-Health and recommender systems \cite{Kaissis2021,Liang2021rec}. 
It has garnered widespread attention for its potential to facilitate the development of privacy-sensitive artificial intelligence (AI) applications.

The primary challenges of enabling FL within wireless networks pertain to the scarcity of radio resources and the substantial energy expenditure required for deploying FL over resource-constrained Internet of Things (IoT) devices \cite{9502547}. These issues are further exacerbated in cloud-edge AI systems, where radio resources and device energy consumption are restricted, and the model updates from individual wireless devices during each round of training become more intricate, potentially encompassing millions to billions of parameters. The total energy consumed in FL over wireless networks depends on the energy consumption during the local training process and communication. Therefore, improving the energy efficiency of FL through wireless networks is of utmost importance to facilitate the practical implementation of environmentally sustainable FL \cite{9134426,he2023reconfigurable,9810113}.

To tackle the problem, various model compression approaches such as sparsification and quantization have been explored in \cite{9042352,9014530,9912341,9562516,9269459,9878377,10.1007/3-540-45465-9_59,pmlr-v97-spring19a,pmlr-v119-rothchild20a}. In particular, \cite{9042352,9014530} proposed an analog aggregation model for FL wherein each local worker sparsifies its local updates, e.g., stochastic gradient, via top-$k$ sparsification and error accumulation, and then employs a measurement matrix to projects the updates to a lower dimensional space. In \cite{9912341}, an one-bit compressive sensing scheme was designed for over-the-air FL to reduce the volume of the transmitted data. Meanwhile, \cite{9562516} exploited the observations of the intrinsic temporal sparsity of local updates in FL and characterized the property using a Markovian probability model. In \cite{9269459}, a transmission strategy over massive MIMO communication system via random permutation was developed, while the authors in \cite{9878377} adopted block sparsification to reduce the computation complexity of gradient reconstruction. The count-sketch data structure \cite{10.1007/3-540-45465-9_59} was also used in sparsification \cite{pmlr-v97-spring19a,pmlr-v119-rothchild20a} for FL to tackle the problem of sparse client participation. 
%In \cite{9148987}, to improve the accuracy after a high ratio sparsification, the authors introduced a general gradient sparsification framework for adaptive optimizers.

Apart from the aforementioned studies that primarily focus on model compression techniques, various works have recently been introduced, focusing on system-level optimization to enhance the energy efficiency of FL 
%from a system-level perspective 
\cite{9807354,9673130,9264742,9953187,10001623,YIN2023,DBLP:journals/monet/HuHY22,DBLP:journals/corr/abs-2111-07911,DBLP:journals/corr/abs-2207-09387}. For example, \cite{9807354} studied the balance between energy consumption and carbon footprints in FL with low-power embedded wireless devices. In \cite{9673130}, the authors formulated and solved two optimization problems concerning the overall learning performance and the energy consumption.
% and vice versa
\cite{9264742} aimed to minimize the total energy consumption subject to a latency constraint while designing parameters including time allocation, bandwidth, power, computation frequency, and learning accuracy. A transmission scheme was designed in \cite{9953187} to minimize the total energy consumption under a quality-of-service constraint over massive multiple-input multiple-output (mMIMO) networks. \cite{10001623} leveraged sparse neural networks (SNN) to further reduce the computation complexity in local training and formulated an energy minimization problem concerning computing, uploading and broadcasting. In \cite{YIN2023}, the authors utilized a game tree to model the FL training process to achieve a desired balance between the model performance, privacy preservation and energy consumption. \cite{DBLP:journals/monet/HuHY22} proposed a resource optimization under a training workload constraint to balance the tradeoff between the energy consumption and model performance. In \cite{DBLP:journals/corr/abs-2111-07911}, the authors focused on balancing the tradeoff between energy efficiency and convergence rate while controlling the optimal precision level. Furthermore, the work in \cite{DBLP:journals/corr/abs-2207-09387} investigated a joint minimization problem whose goal is to minimize the energy consumption with respect to the precision level in both local training and uplink transmission, the number of communication rounds, the number of local training rounds, and the number of participated devices in each communication round. However, the works in \cite{DBLP:journals/corr/abs-2111-07911, DBLP:journals/corr/abs-2207-09387} did not consider FL in cloud radio access network (Cloud-RAN), in which remote radio heads (RRHs) are deployed to capture signals from devices and to forward them to the CS through high-speed, reliable fronthual links \cite{6786060}. The orthogonal frequency division multiple access (OFDMA) based Cloud-RAN has been shown to significantly improve system scalability, throughput, and coverage thanks to its spectral efficiency and low encoding/decoding complexity \cite{DBLP:journals/tcom/LiuBZ15}. In addition, OFDMA-based Cloud-RAN is also suitable for the latest wireless devices equipped with digital modulation chips \cite{stephen2017joint}. 
All the aforementioned works either focus on improving energy efficiency over wireless networks or focus on system optimization in Cloud-RAN without considering FL. 
To the best of our knowledge, there is no existing work that jointly considers the tradeoff between the model performance and energy efficiency while controlling the precision level, wireless power control, and fronthaul rate allocation design in an OFDMA-based Cloud-RAN system. 
%Therefore, the goal of the paper is to formulate an energy minimization optimization problem in OFDMA-based CRAN
%Therefore, in this paper, we design an energy-efficient FL framework in OFDMA-based Cloud-RAN. The proposed FL framework is considered to reduce the energy consumption in the whole training process over Cloud-RAN, which is described as follows.  Based on the local training process and communication over Cloud-RAN, we propose an energy model. Moreover, we derive the convergence property of the framework, and based on the convergence analysis, we develop an alternative optimization approach with respect to limited fronthaul capacity, user power, and the precision levels of QNNs, thus minimizing the overall energy consumption.
In this paper, we aim to address the research gap by designing an energy-efficient FL framework in OFDMA-based Cloud-RAN. Our proposed framework aims to reduce energy consumption during the entire training process over Cloud-RAN.
The main contributions of this paper are summarized as follows:
\begin{itemize}
\item We propose an energy-efficient FL framework in OFDMA-based Cloud-RAN with QNN.
%where each device adopts QNN to train the model and transmits the quantized model updates to the CS. 
To facilitate the proposed FL framework, each device trains its own QNN and transmits the quantized model updates to the RRHs. The RRHs forward these signals to the CS, and then the CS aggregates and generates new global model updates, thereby improving energy efficiency and improving system scalability.
\item 
We  develop a comprehensive energy consumption model that takes into account both the local training at each device and the uplink transmission over Cloud-RAN. By considering these two aspects, our proposed energy consumption model accurately quantifies the total energy consumed in the whole training process.
\item We theoretically analyze the convergence of the proposed FL framework to characterize the impact of the precision levels of QNNs. Guided by the analysis, we formulate an energy minimization problem in OFDMA-based Cloud-RAN to optimize the fronthaul rate allocation, user transmit power allocation, and QNN precision levels while satisfying the limited fronthaul capacity constraint and ensuring the convergence of the proposed FL model to a target accuracy. The formulated problem is then solved by an alternative optimization algorithm.
\item Extensive experiments are conducted to verify the superiority of the proposed FL framework compared to several baselines. We also discover the tradeoff between the optimal precision level, wireless power control, and fronthaul rate allocation. Simulation results demonstrate that our proposed FL framework  significantly reduces the total energy consumption, indicating its potential for achieving a sustainable and environmentally-friendly FL in Cloud-RAN.
\end{itemize}

The rest of the paper is organized as follows. We first introduce the system model of FL with QNN and the communication model over Cloud-RAN in Section \ref{systemmodel}. In Section \ref{energymodel}, we develop the energy consumption model of computation and communication and formulate an energy minimization problem. Sections \ref{solution} and \ref{results} solve the formulated problem using alternative optimization and present extensive numerical results of the proposed algorithm. Finally, we conclude the paper in Section \ref{conclusion}.

%notation
Throughout this paper, $\mathbb{C}$ and $\mathbb{R}$ denotes the complex number set and the real number set, respectively. $ (\cdot)^\transpose $ and $ (\cdot)^\Htranspose $ denotes the operations of transpose and conjugate transpose.
The superscripts   $ \expp{\cdot} $ is the expectation operator. $ \bm{I}_d $ is the identity matrix with $ d \times d $ size. $|x|$ is the operation that take the absolute value of the real value $x$. $\operatorname{sign}(x)$ represents taking the symbol of the real number $ x $.

\section{System model}\label{systemmodel}

\begin{figure}[htb]
\centering
\centerline{\includegraphics[scale=0.35]{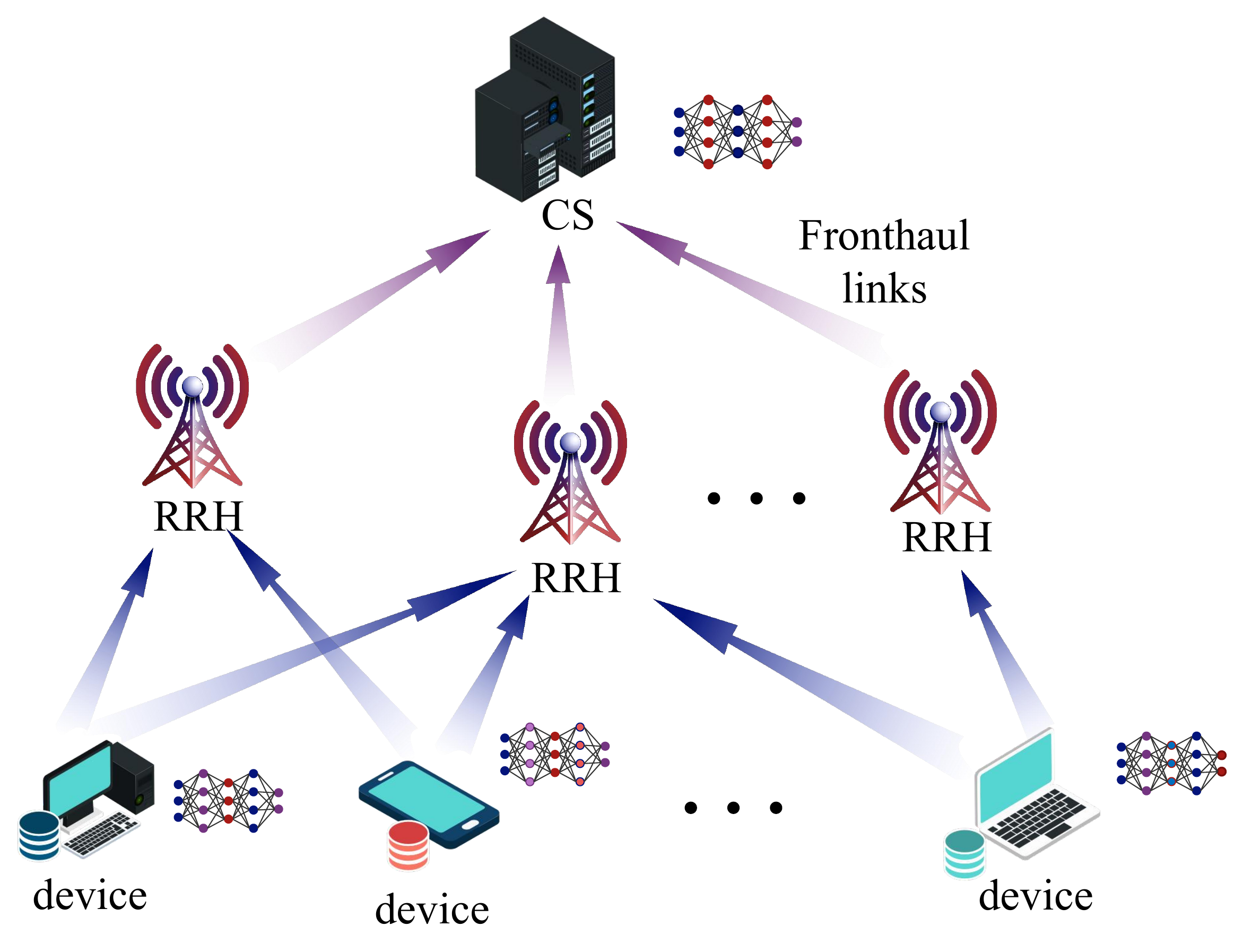}}
\caption{System model of the considered wireless federated learning in Clould-RAN.}
\label{fig:systemmodel}
\end{figure}
As depicted in Fig. \ref{fig:systemmodel}, we consider a wireless FL system in a Cloud-RAN consisting of $K$ devices, $M$ RRHs and one CS. Each RRH is connected to the CS through a fronthaul link with limited capacity. The CS and all local devices collaboratively train a shared learning model. We assume that each device has a local dataset $\mathcal{D}_k = \{  \bm{x}^{ik}, y^{ik}\}_{i=1}^{D_k}$ with $ |\mathcal{D}_k | = D_k $, where $ \{  \bm{x}^{ik}, y^{ik}\}$ denotes the input-output pair for the task with $ \bm{x}^{ik}$ and $y^{ik}$ respectively being the input vector and the corresponding output. The union of all training datasets is denoted as $\mathcal{D}=\bigcup_k\mathcal{D}_k$ with $|\mathcal{D}|=D$ being the total size of data samples. 
For each device $k$, the local loss function $f_k(\bm{w})$ can be defined as
                \begin{equation}
                    f_k(\bm{w}^k) = \frac{1}{{D}_k}\sum_{i\in\mathcal{D}_k}f (\bm{w}^k,  \{  \bm{x}^{ik}, y^{ik}\}),
                \end{equation}
where $f(\bm{w}^k, \{  \bm{x}^{ik}, y^{ik}\})$ is the loss function for parameter $\bm{w}$ on data sample $i$ at device $k$ and $\bm{w}^k$ is the local model trained based on its local dataset $\mathcal{D}_k$. Our goal is to learn a shared model through minimizing the sum of the edge devices’ local loss functions. Specifically, the CS optimizes the 
following global loss function:
%We consider a federated system in Cloud-RAN to support the following distributed optimization task:          
\begin{subequations}\label{global}
    \begin{align}
    &\min\limits_{\bm{w}}    F(\bm{w})  =\frac{1}{D}\sum_{k=1}^{K} \sum_{i=1}^{{D}_k }  f( \bm{w}^k;  \bm{x}^{ik}, y^{ik} ) \quad \text{s.t.} \quad \bm{w}^1 =\cdots= \bm{w}^K=\bm{w},
    \end{align}
    \end{subequations}
    where $\bm{w}\in \mathbb{R}^d$ is the global model parameter vector, and $F(\bm{w})$ is the global loss function of the global shared learning model. To solve problem \eqref{global}, the widely adopted training algorithms, i.e., gradient descent methods, are employed with numerous rounds of communication between the CS and devices. Nevertheless, the transmission of gradients or model parameters may result in considerable communication overhead, particularly when dealing with large models. Furthermore, the energy-constrained devices inherent in IoT systems make it challenging to solve problem \eqref{global} using a power-consuming FL approach.
%The transmission of gradients or model parameters can still bring hugh communication overheads for a large model. However, in practical IoT systems, devices are energy-constrained, which makes it difficult to solve problem \eqref{global} with a power consuming FL process.
% Hence, we aim to design an energy-efficient FL framework that 
%       Hence, we aim to reduce the totol energy consumption which consists of computation energy and communication energy. For optimizing the computation energy, we control the precision level of QNN training. Besides, for optimizing the communication energy, we jointly design the transmit power of each device and the RRHs' fronthaul rate allocation. 
           
\subsection{FL model with Quantized Neural Networks}
To address the aforementioned limitations of FL, this subsection focuses on the design of an energy-efficient FL framework.
%To design an energy-efficient FL framework, this subsection focuses on the design of FL model. 
Specifically, we adopt FedAvg \cite{mcmahan2017communication} to solve problem \eqref{global} and we incorporate QNNs to curtail the computation energy required at each device. In the following, we first specify the neural network deployed at each local device and then present the FL model and FL algorithm. 

Different from many existing works on FL that adopts the conventional 32-bit floating-point format \cite{mcmahan2017communication,DBLP:journals/jsac/ZhengSC21,8952884}, in this work, we adopt a QNN structure to save the computation energy. The QNN structure quantizes the weights and activations in a fixed-point format. In specific, $C_{\text{prec}}$ bits of precision for quantization are used at each device for the training of an identical QNN structure. In this way, the data can be represented more precisely by increasing $C_{\text{prec}}$ at the expense of more energy consumption.
For given $ C_{\text{prec}} $ quantization bits, we consider a
stochastic quantization scheme \cite{DBLP:journals/jsac/WangXSC22} \cite{DBLP:journals/corr/abs-2006-10672}, where any given $w\in\bm{w}$ is quantized as follows
\begin{equation}\label{qnnquanti}
Q(w) = \left\{ 
\begin{aligned}
& \operatorname{sign}(w) \cdot s_{i} \quad \quad&, \text{with probability }  & \frac{s_{i+1} - |w|}{s_{i+1}- s_i}\\
&   \operatorname{sign}(w) \cdot s_{i+1} & ,\text{with probability } &\frac{|w|- s_i}{s_{i+1}- s_i }
\end{aligned}
   \right.
\end{equation}
with
\begin{equation}
s_i = w_{\min} + i \times \frac{ w_{\max} - w_{\min} }{b}, i=0, 1, \dots, b,
\end{equation}
and $  w_{\min} \triangleq \min\{  |\bm{w} | \}$, $ w_{\max} \triangleq  \max \{  |\bm{w} | \}$, $  b  =2^{C_{\text{prec}} }- 1$. Note that $ \{ s_0, s_1, \dots, s_{{b}}\} $ denotes the knobs uniformly distributed in $ [w_{\min}, w_{\max}] $. Based on the aforementioned quantization approach, we derive the quantized weights of layer $s$ for device $k$ as
\begin{equation}
\bm{w}^{Q,k}_{s} = Q(\bm{w}_s^{Q,k}).
\end{equation}
Also, the outputs of layer $s$ is defined as
  \begin{equation}
  o_s = g_s(\bm{w}^{Q,k}_{s} , o_{s-1}^Q ),
  \end{equation}
  where $ g_s(\cdot) $ denotes the operation of layer $s$ on the input, i.e., activation and batch normalization, and $ o_{s-1}^Q  $ represents the quantized output from the preceding layer $s-1$. We further quantize the output $o_s$ and feed it into the next layer as an input. Such approach is essential to all real valued weights projected into the interval $[-1,1]$ during training, i.e, 
         %\begin{equation}
         $\bm{w}^k \triangleq   \operatorname{clip}  (\bm{w}^k, -1,1),$
         %\end{equation}
    where the function $\operatorname{clip} (\cdot,-1,1)$ maps all inputs greater than $1$ to $1$ and all inputs smaller than $-1$ to $-1$, while leaving values within this range unaltered. 
    %the function $\operatorname{clip} (\cdot,-1,1)$ projects each input to $1$ $(-1)$ for any input larger (smaller) than $1$ $(-1)$, while values within the range remain unchanged.
         %Note that to avoid the impact of large-value $\bm{w}^k$ on the quantization \cite{DBLP:journals/jmlr/HubaraCSEB17}, $\bm{w}^k$ is quantized at each training at device $k$.
         %Note that to avoid an excessive increase in $\bm{w}^k$ with no significant impact on quantization . $\bm{w}^k$ is quantized as $\bm{w}^{Q,k}$ at device $k$ after each training.
   Next, we introduce the FL model and the FL algorithm. To learn the global model $\bm{w}$, the SGD algorithm is adopted at each device $k$ to update the local parameter $\bm{w}^k $. In each communication round $t$, the CS randomly selects a subset of devices $\mathcal{K}_t \subseteq \mathcal{K}$ to train the model. Without loss of generality, we assume that a constant number of devices is selected in each round, i.e., $|\mathcal{K}_t|=\bar{K},\forall t \in \{1,\cdots, T\} $. Then, the CS broadcasts the current global model parameter $\bm{w}_t$ to the participated devices via RRHs. After receiving the global model parameter, each device $k$ in $\mathcal{K}_t$ performs the entire training procedure by implementing $I$ steps of SGD as follows
\begin{equation}\label{localsgd}
\bm{w}^{k}_{t,\tau} = \bm{w}^{k}_{t,\tau-1} - \lambda_t  \nabla f_k(\bm{w}_{t,\tau-1}^{Q,k}, \xi^k_{t}),\forall \tau = 1,\cdots,I,
\end{equation}
where $ \bm{w}^{k}_{t,\tau}  $ denotes the local parameter at device $ k $ in the local iteration $\tau $ in the global communication round $ t $, and $ \bm{w}^{Q,k}_{t,\tau-1}  $ is the quantized version of the local parameter of device $k$ in the local iteration $\tau-1 $ in the global communication round. $ \lambda_t$ and $ \xi^k_{t} $ respectively denote the learning rate at global iteration $t$ and the mini-batch at local iteration $\tau$ for device $k$ in $\mathcal{K}_t$. We assume that devices that are not selected in the communication round do not carry out local training. For device $k$, the local model update to be transmitted from device $k$ to the CS after $ I $ rounds of local training at global iteration $ t $ can be given by
\begin{equation}
\bm{d}_{t}^k = \bm{w}^{k}_{t,I} - \bm{w}^{k}_{t},
\end{equation}
where $\bm{w}^{k}_{t}$ denotes the global model parameter in the communication round $t$.
The local model updates contains massive amount of elements, making it impractical to transmit $\bm{d}_{t}^k $ with full precision 
especially for energy-constrained devices. Thus, we utilize the quantization scheme in \eqref{qnnquanti} to $ \bm{d}_{t}^k  $ as follows
       \begin{equation}
       \bm{d}_{t}^{Q,k } = Q(\bm{d}_{t}^k ),
       \end{equation}
      where $ \bm{d}_{t}^{Q,k}  $ is the quantized version of $ \bm{d}_{t}^k  $. After performing the aforementioned operation, the selected devices send their quantized local model update to the CS via RRHs. The CS then updates the global model by aggregating the received model updates from all selected devices as
      \begin{equation}\label{globalupdate}
      \bm{w}_{t+1} = \bm{w}_t + \frac{1}{K} \sum_{k\in\mathcal{K}_t}    \bm{d}_{t}^{Q,k}.
       \end{equation}
       %The communication round is repeated until the global loss function converges to a target accuracy $\epsilon$. 
       The communication round is repeated until the global loss function reaches a desired level of accuracy, denoted by $\epsilon$.  
       We summarize the aforementioned algorithm in Algorithm \ref{alg1}.
\begin{algorithm}[h]
   \caption{ Quantized FL Algorithm via Cloud-RAN}\label{alg1}
   \LinesNumbered
  \KwIn{The number of devices $K$, the number of RRHs $ M $, local running steps $ I $, the number of subcarriers $ N $, initial model parameters $ \bm{w}_0 $, $ t=0 $, target accuracy $ \epsilon $;}
  \Repeat{
  target accuracy $ \epsilon $ is satisfied
  }{The CS randomly selects participated devices and broadcasts $ \bm{w}_t $ to all the devices over error-free channel\;
  Each selected device runs $ I $ steps of SGD according to \eqref{localsgd}\;
  Each selected device transmits $ \bm{d}_t^{Q,k}$ to the CS via Cloud-RAN\;
  The CS updates the global model via \eqref{globalupdate}\;
  $ t\leftarrow t+1 $\;
  }
  \end{algorithm}
\subsection{Communication model}
In this subsection, the communication model is delineated. We consider an OFDMA-based uplink Cloud-RAN system for FL. The system consists of $K$ single-antenna devices and $M$ single-antenna RRHs connected by a wireless link with a total bandwidth of $B$ Hz. The device set and the RRH set are denoted by $\mathcal{K} = \{1, \cdots, K\}$ and $\mathcal{M} = \{1,\cdots, M\}$, respectively. The bandwidth is equally divided into $N$ subcarriers (SCs). The SC set is denoted by $\mathcal{N} = \{1, \cdots, N\}$ where we assume that each SC $n\in N$ belongs to only one device. The set $\Omega_k$ denotes the SCs allocated to device $k,\forall k \in \mathcal{K}$.
%is only allocated to one device.
%with each SC being allocated to a single device
 %We use $\Omega_k$ to represent the set of SCs allocated to device $k$, $\forall k \in \mathcal{K}$. 
 Although dynamic SC allocation can improve the spectral efficiency, in this paper we focus on the joint design of the precision level in QNN, the transmit power of each device, and the fronthaul rate allocation. Thus, for simplicity we assume that the SC allocations among devices, $\Omega_k,\forall k\in\mathcal{K}$ is predetermined without taking dynamic SC allocation into consideration \cite{DBLP:journals/tcom/LiuBZ15}. 

Specifically, at a given communication round $t$, each device $k $ first encodes and modulates each entry of the local updates into $s_{k,n}$, which is an entry of the local model update $ \bm{d}_{t}^{Q,k }  $ of device $k$. Each device then transmits the signal to the RRHs in the uplink. Afterward, the signal received by RRH $m$ is sent to the baseband where it is transformed into parallel signals and demodulated into $N$ streams.
%each RRH $m$ sends the received signal to the baseband, transforms the serial baseband signals to the parallel ones, and demodulates the parallel signals into $N$ streams. 
To simplify the notation, we omit the index of communication round $t$ by writing the transmit symbol $s_{k,n}(t)$ as $s_{k,n} $. Let $p_{k,n}$ be the transmit power of device $k$ at SC $ n $, and $ h_{m,k,n} $ be the channel coefficient from device $ k $ to
RRH $ m $ at SC $ n $. The equivalent baseband complex symbol received by RRH $ m $ at SC $n\in\Omega_k$ can be expressed as
\begin{equation}
y_{m,n} = h_{m,k,n} \sqrt{p_{k,n}} s_{k,n}+ z_{m,n},
\end{equation}
where $z_{m,n} \sim\mathcal{N} (0, \sigma^2_{m,n})$ is the summation of additive white Gaussian noise (AWGN) and out-of-cluster interference at RRH $m$ at SC $ n $ \cite{DBLP:journals/tcom/LiuBZ15}. Without loss of generality, we assume that $z_{m,n}$ is independent over $ m $ and $ n $.
Due to the limited fronthaul capacity, we apply the quantize-and-forward scheme to forward the baseband symbols to the CS \cite{DBLP:journals/tcom/LiuBZ15}. Each RRH first quantizes its baseband signal and then transmits the corresponding digital codewords to the CS. We assume the independence of the symbols received at each RRH as well as the signal quantization across different RRHs. For given the quantization error $ e_{m,n}$ associated to the received signal $  y_{m,n} $ with zero mean and variance $ q_{m,n} $, the baseband quantized symbol of $ \tilde{y}_{m,n} $ is expressed as
      \begin{equation}
      \begin{aligned}
\tilde{y}_{m,n} &=  y_{m,n} + e_{m,n} =  h_{m,k,n} \sqrt{p_{k,n}} s_{k,n} + z_{m,n} + e_{m,n}.
      \end{aligned}
\end{equation}
It is assumed that $  e_{m,n} $ is independent over $m$ and $ n $. 
%After quantizing the symbol, each RRH transforms the parallel encoded bits $ \tilde{y}_{m,n} $ into serial ones and sends them to the CS through the fronthaul link for joint information decoding.
Once the symbols are quantized, each RRH converts the parallel encoded bits $\tilde{y}_{m,n}$ into a serial form and transmits them to the CS through the fronthaul link to enable joint information decoding.
Then, the CS uses the quantization codebooks to recover the baseband quantized symbols $  \tilde{y}_{m,n}   $. 
In this context, we adopt a uniform scalar quantization method which is commonly used in practice \cite{DBLP:journals/tcom/LiuBZ15}. We use this method at each RRH to obtain the sum-rate. A common way to implement uniform quantization is through separate in-phase/quadrature (I/Q) quantization \cite{DBLP:journals/tcom/LiuBZ15}.
%Here, we adopt practical uniform scalar quantization technique \cite{DBLP:journals/tcom/LiuBZ15} at each RRH and derive the achievable sum-rate. A typical method to implement the uniform quantization is via separate in-phase/quadrature (I/Q) quantization \cite{DBLP:journals/tcom/LiuBZ15}. 
Each RRH $m$ uses $C_{m,n}$ quantization bits to perform uniform scalar quantization, leading to $2^{C_{m,n}}$ quantization levels. Under the uniform scalar quantization scheme, the transmission rate from RRH $m$ to the CS in its fronthaul link is expressed as \cite{DBLP:journals/tcom/LiuBZ15}
\begin{equation}
G_m = \sum_{n=1}^{N} G_{m,n} \leq \bar{G}_m,m=1,\cdots, M,
\end{equation}
where $  \bar{G}_m $ is the maximum transmission rate of RRH $ m $. The transmission rate in RRH $m$'s fronthaul link to forward its received data at SC $ n $ is \cite{DBLP:journals/tcom/LiuBZ15}
\begin{equation}
G_{m,n} = \frac{2BC_{m,n}}{N}, \forall m,n,
\end{equation}
where $ B $ is the bandwidth. According to the results of \cite[Proposition 3.2]{DBLP:journals/tcom/LiuBZ15}, the variance $ q_{m,n} $ can be written as
% \begin{equation}
$q_{m,n} = 3 (    |   h_{m,k,n}  |^2p_{k,n} + \sigma_{m,n}^2  ) 2^{-   C_{m,n}}$.
%\end{equation}%\frac{NG_{m,n}}{B}
Thus, the achievable rate for device $ k $ at SC $ n $ is given by
\begin{equation}
\begin{aligned}
& R_{k,n} = \frac{B}{N} \log_2 \left(   1+ \sum_{m=1}^{M} \frac{ |   h_{m,k,n}  |^2p_{k,n}  }{  \sigma_{m,n}^2   +    3 (    |   h_{m,k,n}  |^2p_{k,n} + \sigma_{m,n}^2  ) 2^{-   2C_{m,n}}}     \right).
\end{aligned}
\end{equation}
The corresponding end-to-end transmission rate of device $ k $ is expressed as
\begin{equation}
\begin{aligned}\label{rk}
R_{k}&= \sum_{n\in\Omega_k} R_{k,n}.  \\
\end{aligned}
\end{equation}

\section{ System Energy Consumption Model and problem formulation}\label{energymodel}
In this section, an energy consumption model is proposed for the computation and communication phases of the FL Cloud-RAN system, and an energy minimization problem is formulated to optimize the fronthaul rate allocation, user transmit power allocation, and QNN precision levels subject to constraints on power budget, fronthaul capacity, and learning accuracy.
\subsection{Computing Energy Consumption}\label{energymodelcomp}
   \begin{figure}[htb]
   \centering
   \centerline{\includegraphics[scale=0.35]{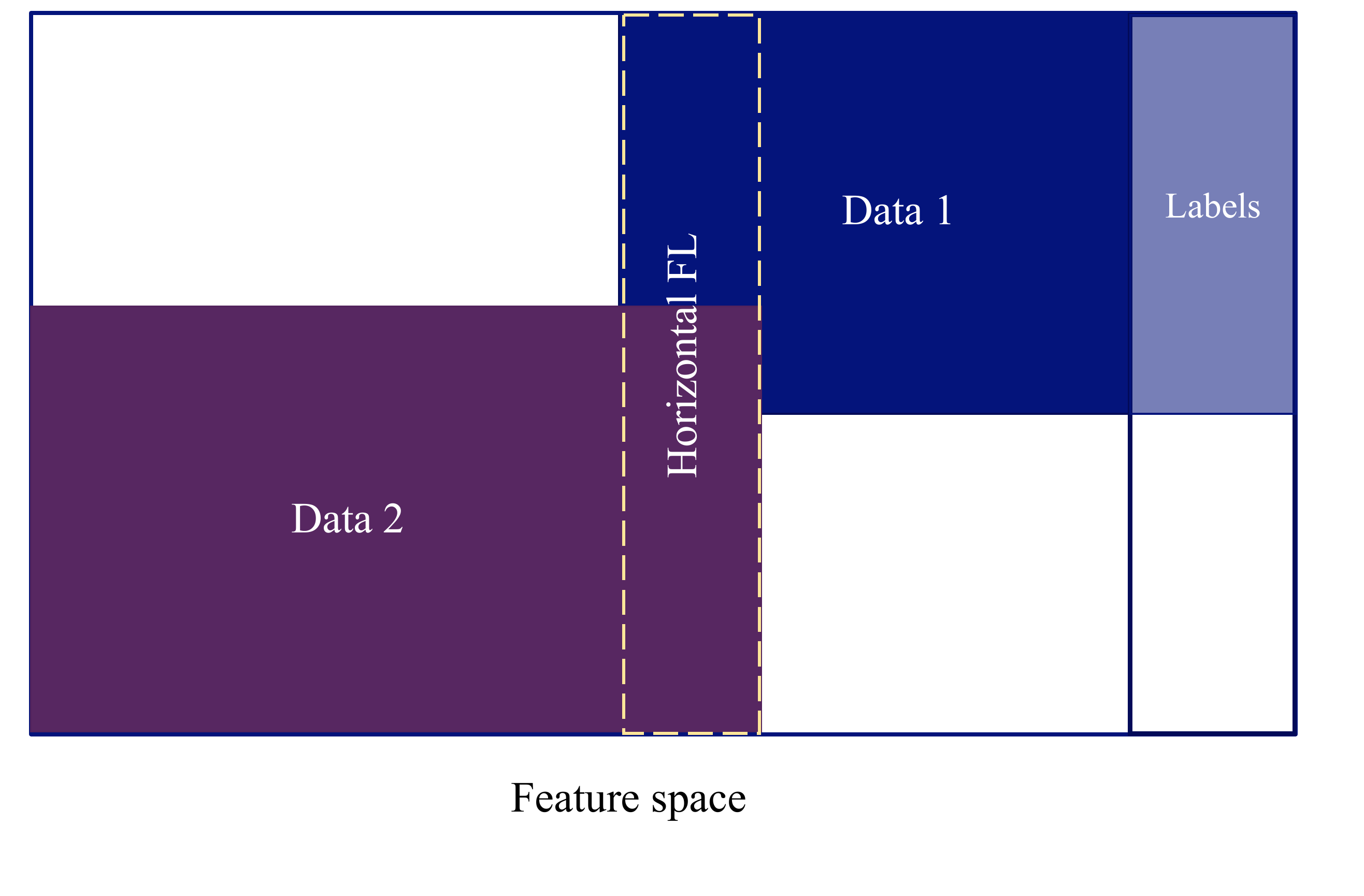}}
  \caption{An illustration of a two dimensional processing chip.}
  \label{fig:qnn}
  \end{figure}
      %In this subsection, 
As shown in Fig. \ref{fig:qnn}, we consider a typical two-dimensional processing chip \cite{DBLP:journals/corr/abs-2111-07911} \cite{DBLP:journals/corr/abs-2207-09387} where the chip consists of a parallel neuron array, $ u $ multiply-accumulate (MAC) units, a main buffer and a local buffer. Specifically, the weights and activations of the current layers are stored in the main buffer, whereas those of the used weights and activations are cached in the local buffer. The energy consumption model of a MAC operation with the precision level of quantization $ C_{\text{prec}}$ can be expressed as $E_{\operatorname{MAC}}(C_{\text{prec}}) = A\left (\frac{C_{\text{prec}}}{C_{\text{max}}}  \right)^\alpha$,	 
where $ A >0$, $ 1<\alpha<2 $, and $ C_{\text{max}} $ denotes the maximum precision level \cite{DBLP:conf/acssc/MoonsGBV17} \cite{Moons2019}. 
%The energy consumption for accessing a local buffer $ E_{lb} $ can be modeled as $   E_{\operatorname{MAC}}(C_{\text{prec}}) $. Also, the energy consumption for accessing a main buffer is $ E_{mb} = 2  E_{\operatorname{MAC}}(C_{\text{prec}})  $.
The energy consumption required to access a local buffer is denoted by $ E_{lb} $, which can be modeled as $ E_{\operatorname{MAC}}(C_{\text{prec}}) $. Also, the energy consumption for accessing a main buffer is denoted as $ E_{mb}=2E_{\operatorname{MAC}}(C_{\text{prec}}) $. Thus, at communication round $ t $, the energy consumption of device $k$ for one iteration of local training $ E^{\text{C},k}(C_{\text{prec}}) $ can be expressed as \cite{DBLP:conf/acssc/MoonsGBV17} \cite{Moons2019} 
 \begin{equation}
E^	{\text{C},k}(C_{\text{prec}})  = E_{C}(C_{\text{prec}}) + E_{W}(C_{\text{prec}}) + E_A(C_{\text{prec}}),
     \end{equation}
  where 
  \begin{itemize}
\item $ E_{C}(C_{\text{prec}}) =   E_{\operatorname{MAC}}(C_{\text{prec}}) N_{\operatorname{MAC}} + 3O_s  E_{\operatorname{MAC}}(C_{\text{max}}) $ denotes the computing energy, $  N_{\operatorname{MAC}} $ is the number of MAC operations and $ O_s $ represents the number of intermediate
outputs throughout the network \cite{DBLP:journals/jmlr/HubaraCSEB17}. $ E_{C}(C_{\text{prec}})  $ is the summation of energy consumed by batch normalization, activation, and bias.
\item $  E_{W}(C_{\text{prec}})  = E_{mb}    N_{w} + E_{lb}  N_{\operatorname{MAC}}   \sqrt{   \frac{  C_{\text{prec}   }}{b C_{\text{max}}  }}$ denotes the access energy for obtaining weights from buffers where $  N_{w} $ represents the number of weights,
\item $ E_A(C_{\text{prec}})  = 2E_{mb}O_s + E_{lb}  N_{\operatorname{MAC}} \sqrt{   \frac{  C_{\text{prec}   }}{b C_{\text{max}}  }}$ denotes the access energy for obtaining activations from the buffers.
  \end{itemize}  
  Thus, at one communication round $t$, the total computing energy consumption for all selected devices is the sum of computing energy consumed by the $ I $ rounds of local model updates, which is given by
  \begin{equation}
  E^{C, t} = I\sum_{k\in\mathcal{K}_t} E^{C,k}(C_{\text{prec}}) .
  \end{equation}

\subsection{Communication Energy Consumption}\label{transenergymodel}
In this subsection, an energy consumption model is proposed for communication via OFDMA-based Cloud-RAN, considering the transmission latency and the transmit power of devices and the fronthaul capacity.

First, we divide the whole uplink transmission latency into two parts, i.e. the latency caused by sending the local training updates from devices to the RRHs over the wireless channel and the latency due to the transmission from the RRHs to the CS via the fronthaul links. The upload latency from the device $ k $ to the RRHs over the wireless channel is given by \cite{DBLP:journals/corr/abs-2202-07775}, \cite{DBLP:conf/icassp/LiLL18} 
       \begin{equation}\label{tk1}
        \begin{aligned}
       T_{k,1}& = \frac{\| \bm{d}^{Q,k}_t\|}{R_k} = \frac{  \| \bm{d}^{k}_t   \| C_{\text{prec}} }{R_kC_{\text{max}}   }  ,
        \end{aligned}
       \end{equation}
       where $ R_k $ is defined in \eqref{rk}.
       The fronthaul latency resulting from the forwarding signals from the RRH $ m $ to the CS over the fronthaul link is given by \cite{DBLP:journals/corr/abs-2202-07775}, \cite{DBLP:conf/icassp/LiLL18} 
       \begin{equation}\label{tk2}
                T_{k,m,2} = \frac{ 2\| {\bm{d}}^{Q,k}_t \| C_{m,k} }{ G_{m} } = \frac{2  \| \bm{d}^{k}_t   \| C_{\text{prec}} C_{m,k}}{  G_{m}  C_{\text{max}}  } ,
        \end{equation}
        where $ C_{m,k} = \sum_{n\in\Omega_k}C_{m,n} $ denotes the number of bits used to quantize both real and imaginary parts of the uplink signal for the SCs in $ \Omega_k $.

        Similarly, the total power consumption is divided into two parts: the power consumed by the users and the fronthaul link. Let $p_k$ represent the transmit power consumption of device k, and $P_m, m\in\mathcal{M}$ represent the power consumed by the $m$-th RRH \cite{DBLP:journals/tgcn/NgoTDML18} \cite{DBLP:journals/twc/MasoumiE20}.		
    For simplicity, we omit the amplifier and the circuit power consumed by the RRHs, which can be considered to be independent with the fronthaul rate \cite{DBLP:journals/tgcn/NgoTDML18}.
       Thus, at communication round $ t $, the total communication energy consumption $ E^{\text{trans}}_{\mathcal{K}_t} $ at all active devices and all RRHs are given as follows
       \begin{equation}\label{E_Kt_trans}
       \begin{aligned}
    &	E^{\text{trans}}_{\mathcal{K}_t}  = &\underbrace{ \sum_{k\in\mathcal{K}_t} \left(  T_{k,1}\times p_k\right) }_{\text{energy consumed by devices}} +  \underbrace{  \sum_{m=1}^{M} \left(  \sum_{k\in\mathcal{K}_t} T_{k,m,2} \right) \times P_m  }_{\text{energy consumed by the RRHs}},
         \end{aligned}
       \end{equation} 
       where $ P_m = B(\sum_{n=1}^{N}\frac{2C_{m,n}}{N} )  P_{fl, m} = G_m    P_{fl, m}  $ and $ P_{fl,m} $ denotes the traffic-dependent power (in Watt/bit/s).

\subsection{Problem Formulation}

In this work, we aim at minimizing the total computation and communication energy consumption while ensuring the convergence accuracy $\epsilon$. According to \cite{DBLP:journals/corr/abs-2111-07911}, a tradeoff exists between the energy consumption and the convergence rate with respect to the level of precision $ C_{\text{prec}}$ of QNN, quantization bits of fronthaul links $ C_{m,n} $, and transmit power of devices $ p_{k,n} $. Hence, determining the optimal values of $ C_{\text{prec}} $, $ C_{m,n} $, and $ p_{k,n} $ is crucial to enhance the tradeoff and achieve the target accuracy. The formulated problem is given as
\begin{subequations}
\label{27}
\begin{align}
\label{wsr}
\underset{C_{\text{prec}} , \{ C_{m,n},p_{k,n} \}}{\operatorname{min}}&  \expp{  \sum_{t=1}^{T }  \sum_{k\in\mathcal{K}_t}  I E^{\text{C},k}(C_{\text{prec}})  + \sum_{t=1}^T E^{\text{trans}}_{\mathcal{K}_t}  } \\
\label{cons_quan_qnn}   \text { s.t. } \quad &  C_{\text{prec}} \in [1,\cdots, C_{\text{max}}] ,& \\ 
\label{cons_f}    & \expp{F(\bm{w}_T)} - F(\bm{w}^*) \leq \epsilon , \\
\label{cons_power}  &\sum_{n\in\Omega_k} p_{k,n} \leq \bar{P}_k, \forall k \in \mathcal{K}_t,t \in [T], \\
\label{cons_fronthaullink} &\sum_{n=1}^{N} C_{m,n} \leq \bar{C}_m, \forall m\in \mathcal{M},  \\
%\label{front_G} &G_{m,n} = \frac{2B C_{m,n}}{N}, \forall m\in \mathcal{M}, \forall n \in \mathcal{N} ,  \\
\label{front_C} &C_{m,n} \in \mathbb{N}^+ , \forall m\in \mathcal{M}, \forall n \in \mathcal{N},
\end{align}
\end{subequations}
where $ I $ denotes the number of local iterations, $ \expp{F(\bm{w}_T)}  $ denotes the expectation of global loss function $ F $ after $ T $ global iterations, $  \bar{P}_k $ is the maximum transmit power for device $ k $, and $   F(\bm{w}^*) $ denotes the minimum value of $ F $ with optimal parameters $  \bm{w}^*$. Note that constraint \eqref{cons_quan_qnn} restricts the feasible range of $ C_{\text{prec}} $,  \eqref{cons_f} restricts the target accuracy of convergence, \eqref{cons_power} defines the constraint on the maximum transmit power for each device, constraint \eqref{cons_fronthaullink} is the fronthaul capacity limit, and constraint \eqref{front_C} specifies that the number of quantization bits are positive integers.

\section{ PROPOSED APPROACH FOR ENERGY-EFFICIENT FL in Cloud-RAN}\label{solution}

In this section, we first analyze the relationship between the convergence speed and convergence accuracy.
%we obtain the analytical relationship between \eqref{cons_f} and $ C_{\text{prec}} , \{ C_{m,n},p_{k,n} \} $ to derive $ T $ with respect to $ \epsilon $. 
%Based on the energy minimization problem presented above, 
Then, we propose to solve problem \eqref{27} by the alternating optimization approach. Specifically, the transmit power $ \{p_{k,n}\},\forall k\in\mathcal{K}_t, n\in\Omega_k $, the precision level for quantization
$ \{C_{\text{prec}} \}$ in QNN, and the quantization bits allocation $ \{C_{m,n} \},\forall m,n$ for Cloud-RAN
are optimized in an alternative manner until the algorithm
converges.
\subsection{Convergence analysis}\label{assum}
To solve problem \eqref{27}, we first characterize the relationship between the convergence speed $T$ and convergence accuracy $ \epsilon $.
%discuss the analytical relationship between \eqref{cons_f} and $ C_{\text{prec}} , \{ C_{m,n},p_{k,n}  \} $ to derive $ T $ with respect to $ \epsilon $. 
We provide five assumptions in this subsection for the convergence analysis. Assumption 1 - 5 are fairly standard and have been widely used in the convergence analysis \cite{pmlr-v54-mcmahan17a}. Assumption \ref{ass5} sets an upper bound on the maximum value of the weight, ensuring that the quantization error can be bounded properly \cite{DBLP:journals/jsac/ZhengSC21}.

\begin{assumption}\label{ass1}
The local loss function $ f_k(\bm{w}), \forall k\in \mathcal{K} $ is $L $-smooth: for all $ \bm{x} $ and $ \bm{y} $, we have $ f_k(\bm{\bm{x}}) \leq f_k(\bm{y}) + (\bm{x} - \bm{y})^\transpose + \frac{L}{2}\norm{\bm{x} - \bm{y}}_2^2$.
\end{assumption}

\begin{assumption}\label{ass2}
The local loss function $f_k(\bm{w}), \forall k\in \mathcal{K}   $ is $ \mu $-strongly convex: for all $ \bm{x} $ and $ \bm{y} $, we have $   f_k(\bm{\bm{x}}) \geq f_k(\bm{y}) + (\bm{x} - \bm{y})^\transpose + \frac{\mu}{2}\norm{\bm{x} - \bm{y}}_2^2$.
\end{assumption}

\begin{assumption}\label{ass3}
The variance of local gradients is bounded by $ \expp{\norm{\nabla f_k(\bm{w}_t^k, \zeta^k_t) -\nabla f_k( \bm{w}_t^k) }^2} \leq \sigma_k^2,\forall k\in \mathcal{K}, t\in\{1, \cdots, T\} $, where $   \zeta^k_t$ is sampled from $ k $-th device's local data uniformly at random.
\end{assumption}

\begin{assumption}\label{ass4}
The expected squared norm of local gradients is uniformly bounded, i.e., 

$ \expp{\norm{\nabla f_k(\bm{w}_t^k, \zeta^k_t) }^2} \leq G^2,\forall k\in \mathcal{K}, t\in\{1, \cdots, T\}  $.
\end{assumption}

\begin{assumption}\label{ass5}
$ \max \norm{\bm{w}_t^k}_{\infty} \leq M , \forall k \in \mathcal{K},t \in \{1, \cdots, T\}$, for constant $ M > 0 $. 
\end{assumption}

We propose Lemma \ref{lemma1} to characterize the features of the aforementioned stochastic quantization method.
%in \eqref{qnnquanti}.
\begin{lemma} \label{lemma1}
For the stochastic quantization $  Q(\cdot) $, a scalar value
% \begin{equation}
% \begin{aligned}
$\expp{Q(\bm{w})} = \bm{w}$,
% \end{aligned}
% \end{equation}
and the associated quantization error can be bounded by 
% \begin{equation}
% \begin{aligned}
$\expp{\norm{ Q(\bm{w}) - \bm{w}}^2} \leq \frac{d \varepsilon \norm{\bm{w}}^2}{4 b^ 2}$,
% \end{aligned}
% \end{equation}
where $ \varepsilon = \frac{| w_{\max}  -  w_{\min} |^2}{  \norm{\bm{w}}^2  } $, $ 0 \leq \varepsilon \leq 1 $, and $ b =2^{C_{\text{prec}} }- 1   $.
\end{lemma}
\begin{proof}
The complete proof can be found in Appendix \ref{app1}.
\end{proof}
From Lemma \ref{lemma1}, we could infer that the quantization scheme is unbiased since the expectation is zero. However, the quantization error increases with large model parameters.

\theorem  \label{the1} With given learning rate $  \lambda_t = \frac{\beta}{t+\gamma}$, $ \beta > \frac{1}{\mu} $, we have
\begin{equation}
\expp{ F({\bm{w}}_T)} - F(\bm{w}^*) \leq   \frac{L}{2(\gamma+TI)}  \frac{\beta^2 D }{\beta \mu -1} ,
\end{equation}
where $ D $ is defined as
\begin{equation}
\begin{aligned}
D&= \sum_{k=1}^{K} \frac{\sigma_k^2}{K^2} + \frac{d \varepsilon M	^2}{4 (2^{C_{\text{prec}}} -1)^ 2}  (1- \mu)+ 4  (I-1)^2G^2+\frac{4d \varepsilon I ^2 G^2}{\bar{K}  4(2^{C_\text{prec}} - 1) ^2  }   +  \frac{ 4(K-\bar{K}) }{ \bar{K}(K-1) }I^2G^2.
\end{aligned}
\end{equation}
\begin{proof}
The complete proof can be found in Appendix \ref{app3}.
\end{proof}
\begin{remark}
From Theorem \ref{the1} and the constraint in \eqref{cons_f}, we derive the following inequality by upper bounding the error gap by $\epsilon$:
\begin{equation}
\expp{ F({\bm{w}}_T)} - F(\bm{w}^*) \leq \frac{L}{2(\gamma+TI)}  \frac{\beta^2 D }{\beta \mu -1} \leq \epsilon.
\end{equation}
For given other fixed parameters, we can characterize the relationship between the total communication rounds $T$ and $\epsilon$ through taking the equality as \cite{DBLP:journals/corr/abs-2111-07911}
\begin{equation}\label{eps_t_rela}
T=\frac{L} {2I\epsilon}\frac{\beta^2 D }{\beta \mu -1}-\frac{\gamma}{I}.
\end{equation}
Note that the relationship between $T$ and $\epsilon$ plays an important role to simplify and solve the optimization problem in Section \ref{optqnn} and \ref{optcranpower}.
\end{remark}

\subsection{Optimizing level of precision for QNN}\label{optqnn}
In this subsection, we fix the transmit power $p_{k,n}$ and quantization bits allocation $C_{m,n}$ for Cloud-RAN and optimize the quantization precision level for QNN by solving the following problem:
\begin{subequations}
\label{qnn1}
\begin{align}
\label{obj}
\underset{C_{\text{prec}}}  {\operatorname{min}}\quad &  \expp{ \sum_{t=1}^{T }  \sum_{k\in\mathcal{K}_t}  I E^{\text{C},k}(C_{\text{prec}})  +  \sum_{t=1}^T E^{\text{trans}}_{\mathcal{K}_t}  } \\
\text { s.t. } \quad &   \expp{F(\bm{w}_T)} - F(\bm{w}^*) \leq \epsilon , \quad C_{\text{prec}} \in [1,\cdots, C_{\text{max}}],
\end{align}
\end{subequations}
%where $  \expp{F(\bm{w}_T)} $ denotes the expectation of the global loss function after $ T $ communication rounds, $ I $ denotes the number of local steps, $F(\bm{w}^*)  $ denotes the minimum value of $ F $ over optimal parameter $\bm{w}^*  $, $ \epsilon $ denotes the target accuracy. 
Note that we assume that $ K_t $ devices are randomly selected at iteration $ t ,\forall t\in\{1, \cdots, T\}$, the expectation of the objective function of \eqref{obj} can be derived as 
\begin{equation}
\begin{aligned}
\expp{ \sum_{t=1}^{T }  \sum_{k\in\mathcal{K}_t}  I E^{\text{C},k}(C_{\text{prec}})  +  \sum_{t=1}^T E^{\text{trans}}_{\mathcal{K}_t}  } =\frac{\bar{K} T}{K}  \left(   \sum_{k=1}^{K}   I E^{\text{C},k}(C_{\text{prec}})   +   E^{\text{trans}}_{\mathcal{K}}  \right) .
\end{aligned}
\end{equation}
Based on the relationship between $T$ and $\epsilon$ specified in \eqref{eps_t_rela}, the problem is therefore approximated as
%\nonumber &
\begin{subequations}
\label{obj_C_prec}
\begin{align}
\label{obj1}
\underset{C_{\text{prec}}}  {\operatorname{min}}\quad &   \frac{\bar{K}}{K} \left(\frac{L} {2I\epsilon}\frac{\beta^2 D }{\beta \mu -1}-\frac{\gamma}{I} \right) \left(   {  \sum_{k=1}^{K}   I E^{\text{C},k}(C_{\text{prec}})   +   E^{\text{trans}}_{\mathcal{K}}  }  \right)   \triangleq f^E(C_{\text{prec}})\\ 
\label{c_pre}   \text { s.t. } \quad &  C_{\text{prec}} \in [1,\cdots, C_{\text{max}}] ,&
\end{align}
\end{subequations}
where $E^{\text{trans}}_{\mathcal{K}}   $ is defined in \eqref{E_Kt_trans}.
% 	 \begin{equation}
% 	 \begin{aligned}
% 	  		&	E^{\text{trans}}_{\mathcal{K}} \\
% 	  			= &\underbrace{  \sum_{k=1}^K  \left(  T_{k,1}\times p_k\right) }_{\text{users' transmit energy}} +  \underbrace{  \sum_{m=1}^{M}\sum_{k=1}^K T_{k,m,2} \times  (G_m P_{fl, m}   )  }_{\text{fronthaul links' transmit energy}}.\\
% 	 \end{aligned}
% 	 \end{equation}
Note that any optimal solution $ C_{\text{prec}}^\star $ of problem \eqref{obj1} can satisfy problem \eqref{obj} \cite{9488679}.
To solve problem \eqref{obj_C_prec}, we first relax $C_{\text{prec}}  $ as a continuous variable, which will be rounded back to an integer value later. We observe that, as $ C_{\text{prec}} $ increases, $ E^{\text{C},k}(C_{\text{prec}}) $ and $  E^{\text{trans}}_{\mathcal{K}_t} $ becomes larger while $ T $ becomes smaller with fixed $ \{p_{k,n}\} $ and $ \{C_{m,n}\} $. Hence, we can minimize the objective function in \eqref{obj1} to find the local optimal $ C_{\text{prec}}^* $. Due to the fact that the objective function \eqref{obj1} is differentiable with respect to $ C_{\text{prec}} $ in the given range, we can find $ C_{\text{prec}}^* $ by solving $ \frac{\partial  f^E(C_{\text{prec}})}{\partial C_{\text{prec}}} =0$ according to Fermat's Theorem \cite{DBLP:books/sp/BauschkeC11}. Since it is an unconstrained problem and difficult to derive $ C_{\text{prec}}^*  $ analytically, we can obtain the local optimal $  C_{\text{prec}}^\star $ numerically by a line search method \cite{DBLP:books/sp/NocedalW99}.

\subsection{Optimizing the transmit power and quantization bits allocation for Cloud-RAN}\label{optcranpower}
In this subsection, we fix the quantization precision level for QNN and optimize the transmit power and quantization bits allocation for Cloud-RAN by solving the following problem:
\begin{subequations}
\label{cran1}
\begin{align}
\label{obj2}
\underset{\{C_{m,n},p_{k,n}\} }{\operatorname{min}}&  \expp{  \sum_{t=1}^{T }  \sum_{k\in\mathcal{K}_t}  I E^{\text{C},k}(C_{\text{prec}})  + \sum_{t=1}^T E^{\text{trans}}_{\mathcal{K}_t}  } \\
\text { s.t. }  & \expp{F(\bm{w}_T)} - F(\bm{w}^*) \leq \epsilon , \sum_{n\in\Omega_k} p_{k,n} \leq \bar{P}_k, \forall k \in \mathcal{K}_t, \\
&\sum_{n=1}^{N} C_{m,n} \leq \bar{C}_m, \forall m\in \mathcal{M},  C_{m,n} \in \mathbb{N}^+ , \forall m\in \mathcal{M}, \forall n \in \mathcal{N} .
\end{align}
\end{subequations}
Similarly, we adopt \eqref{eps_t_rela} and set $ T=\frac{L} {2I\epsilon}\frac{\beta^2 D }{\beta \mu -1}-\frac{\gamma}{I} $ and approximate the problem \eqref{cran1} as
\begin{subequations}
\label{cran2}
\begin{align}\label{obj_E_trans}
\underset{\{C_{m,n},p_{k,n}\} }{\operatorname{min}}&  \frac{K_t}{K} \left(  \frac{L} {2I\epsilon}\frac{\beta^2 D }{\beta \mu -1}-\frac{\gamma}{I}  \right) {   E^{\text{trans}}_{\mathcal{K}}  }   \\
\text { s.t. }  &\sum_{n\in\Omega_k} p_{k,n} \leq \bar{P}_k, \forall k \in \mathcal{K}_t,\\
&\sum_{n=1}^{N} C_{m,n} \leq \bar{C}_m, \forall m\in \mathcal{M},  C_{m,n} \in \mathbb{N}^+ , \forall m\in \mathcal{M}, \forall n \in \mathcal{N}
\end{align}
\end{subequations}
Since the objective function \eqref{obj_E_trans} is not convex over $ C_{m,n} $ with given $ p_{k,n} $, we transform the design variables $C_{m,n}$ in problem \eqref{cran2} into 
% \begin{equation}
$\psi_{m,n} = 2^{\frac{NG_{m,n}}{B}} = 2^{2C_{m,n}}$.
% \end{equation}
The term $ E^{\text{trans}}_{\mathcal{K}}  $ can be transformed into the following form:
\begin{equation}\label{phi_first}
\begin{aligned}
E^{\text{trans}}_{\mathcal{K}}=  &\sum_{k=1}^{K}  \frac{  \| \bm{d}^{k}_t   \| C_{\text{prec}} p_k /C_{\text{max}}  }{ \sum_{n\in\Omega_k}     \frac{B}{N} \log_2 \left(   1+ \sum_{m=1}^{M} \frac{ |   h_{m,k,n}  |^2p_{k,n} \psi_{m,n} }{  \sigma_{m,n}^2   \psi_{m,n} +        |   h_{m,k,n}  |^2  p_{k,n} + \sigma_{m,n}^2    }\right)   }  \\
&+  \sum_{m=1}^{M}\sum_{k=1}^{K}   {  \| \bm{d}^{k}_t  \| P_{fl, m} C_{\text{prec} } / C_{\text{max}} \sum_{n\in\Omega_k} \log_2 \psi_{m,n}  }, \\
\end{aligned}
\end{equation}
and problem \eqref{cran2} can be rewritten as 
\begin{subequations}
\label{cran3}
\begin{align}
\label{obj2_2variable}
\underset{\{\psi_{m,n},p_{k,n} \}}{\operatorname{min}}  &\quad  {   E^{\text{trans}} _{\mathcal{K}}  }   \\
\text { s.t. }   & \sum_{n\in\Omega_k} p_{k,n} \leq \bar{P}_k, \forall k \in \mathcal{K}, \sum_{n=1}^{N} C_{m,n}\leq \bar{C}_m,\forall m,\\
&\label{inte_con}\psi_{m,n}=2^{2C_{m,n}},C_{m,n}\in\{1,2,\cdots\}, \forall m,n.
\end{align}
\end{subequations}
Although $ E^{\text{trans}}_{\mathcal{K}} $ is equivalently transformed into \eqref{phi_first}, the optimization problem \eqref{cran3} is still a non-convex problem with respect to $ \{p_{k,n},\psi_{m,n} \}  $. In the following, we propose a two-stage algorithm to solve problem \eqref{cran3} based on alternating optimization and convex approximation. Specifically, we first fix the transmit power to optimize quantization bits allocation for Cloud-RAN and then fix quantization bits allocation for Cloud-RAN to optimize the remaining variables:
\subsubsection{fix $\{ \psi_{m,n}\} $ to solve $ \{p_{k,n} \}$}

First, by fixing $ \psi_{m,n} $, we optimize the transmit power allocation $ p_{k,n} $ by solving the following problem:
\begin{subequations}
\label{obj_p_all}
\begin{align}
\underset{\{p_{k,n} \}}{\operatorname{min}}\quad  {   E^{\text{trans}}_{\mathcal{K}}   }  \quad
\text { s.t. }  \sum_{n\in\Omega_k} p_{k,n} \leq \bar{P}_k, \forall k \in \mathcal{K}.
\end{align}
\end{subequations}
To deal with problem \eqref{obj_p_all}, we apply fractional programming (FP) \cite{DBLP:journals/tsp/ShenY18}. 
Specifically, we first define 
\begin{equation}\label{A_k}
\begin{aligned}
&A_k(p_{k,n})= \sum_{n\in\Omega_k}     \frac{B}{N} \log_2 \left(   1+ \sum_{m=1}^{M} \frac{ |   h_{m,k,n}  |^2p_{k,n} \psi_{m,n} }{  \sigma_{m,n}^2   \psi_{m,n} +        |   h_{m,k,n}  |^2  p_{k,n} + \sigma_{m,n}^2    }\right)  
\end{aligned}
\end{equation}
and
\begin{equation}\label{B_k}
B_k(p_{k,n})= \frac{\| \bm{d}^{k}_t   \| C_{\text{prec}}} {C_{\text{max}} } \sum_{n\in\Omega_k} p_{k,n}.
\end{equation}
Then, taking the multiplicative inverse of the objective function in problem \eqref{obj_p_all}, we approximate the original energy minimization problem as
\begin{subequations}\label{max_p}
\begin{align}
\underset{\{p_{k,n} \}}{\operatorname{max}}  & \quad \frac{1}{\sum_{k=1}^{K}  \frac{ B_k(p_{k,n})   }{ A_k(p_{k,n})   }  } \quad
\text { s.t. }   \sum_{n\in\Omega_k} p_{k,n} \leq \bar{P}_k, \forall k \in \mathcal{K}.
\end{align}
\end{subequations}
With Cauchy-Schwarz inequality $ \frac{1}{\sum_{i=1}^n \frac{1}{x_i}} \leq \frac{\sum_{i=1}^n x_i}{n} $, we derive that
\begin{equation} \label{cauchyine}
\frac{1}{\sum_{k=1}^{K}  \frac{ B_k(p_{k,n})   }{ A_k(p_{k,n})   }  } \leq \frac{\sum_{k=1}^{K}   \frac{ A_k(p_{k,n})   }{ B_k(p_{k,n})    } }{K}.
\end{equation}
%Applying \eqref{cauchyine} and discarding constant coefficient, we have the following maximization optimization problem
Substituting \eqref{cauchyine} into \eqref{max_p},
\begin{subequations}\label{FP_ori}
\begin{align}
\label{FP_ori_a}\underset{\{p_{k,n} \}}{\operatorname{max}}  \quad  \sum_{k=1}^{K} \frac{   A_k(p_{k,n}) }{B_k(p_{k,n})}  \quad 
\text { s.t. }\sum_{n\in\Omega_k} p_{k,n} \leq \bar{P}_k, \forall k \in \mathcal{K}.
\end{align}
\end{subequations}
Problem \eqref{FP_ori} is a routine fractional programming problem \cite{DBLP:journals/tsp/ShenY18}.
By reformulating \eqref{FP_ori_a} based on the quadratic transform \cite{DBLP:journals/tsp/ShenY18}, we derive the following optimization problem
\begin{subequations}\label{FP_qua}
\begin{align}
\underset{\{p_{k,n}  \}, \{y_k\}}{\operatorname{max}}  & \sum_{k=1}^{K}  \left(     2y_k \sqrt{ A_k(p_{k,n})} - y_k^2 B_k(p_{k,n}) \right) \\
\text { s.t. } & \sum_{n\in\Omega_k} p_{k,n} \leq \bar{P}_k, y_k \in \mathbb{R}, \forall k \in \mathcal{K}.
\end{align}
\end{subequations}
Note that the numerator $   A_k(p_{k,n})$ is a concave function over $ \{p_{k,n}\} $ and the denominator $  B_k(p_{k,n}) $ is a convex function over $ \{p_{k,n} \}$ \cite{DBLP:journals/tcom/LiuBZ15}.
When $ p_{k,n} $ is fixed, the optimal $ y_k $ can be obtained as \cite{DBLP:journals/tsp/ShenY18}
\begin{equation}\label{y_opt}
y_k^\star = \frac{\sqrt{ A_k(p_{k,n})}}{ B_k(p_{k,n})}, \forall k \in \mathcal{K}_t.
\end{equation}
%concavity  convexity
When $ y_k $ is fixed, due to the fact that each $ A_k(p_{k,n}) $ and $ B_k(p_{k,n}) $ are respectively concave and convex, and that the square-root function is concave and increasing, the quadratic transform $2y_k \sqrt{ A_k(p_{k,n})} - y_k^2 B_k(p_{k,n})  $
is concave in $ \{p_{k,n}\} $ for fixed $ y_k $. Therefore, we can solve the problem using CVX toolbox \cite{cvx} \cite{gb08}. The overall algorithm for solving problem \eqref{obj_p_all} is given in Algorithm \ref{algFP}.
%%%% algorithm

\begin{algorithm}[h]
\caption{ Iterative approach for solving problem \eqref{obj_p_all}}\label{algFP}
\LinesNumbered
\KwIn{$ C_{prec},\{ C_{m,n}\} $;}
\Initialization{Set $ \{p_{k,n}\} $ to a feasible value    }\;
Transform problem \eqref{obj_p_all} into a quadratic FP problem \eqref{FP_qua}\;
\Repeat{ convergence}
{
Update $ y_k$ by \eqref{y_opt}\;
Fix $ y_k $ and update $ \{p_{k,n}\} $ by solving the reformulated problem \eqref{FP_qua} over $ \{p_{k,n}\} $\;
}
\KwOut{  $ \{p_{k,n}\}$ }
\end{algorithm}

\subsubsection{fix $ p_{k,n} $ to solve $ \psi _{m,n}$}
\begin{algorithm}[h]
\caption{ Iterative approach for solving problem \eqref{phi_all}}\label{alg_phi}
\LinesNumbered
\KwIn{$ C_{prec},\{ p_{k,n}\} $;}
\Initialization{Set $ \{\tilde{\psi}_{m,n}\} $ to a feasible value\;    }
\Repeat{ convergence}
{
Fix $  \{\tilde{\psi}_{m,n}\}  $ and update $  \{{\psi}_{m,n}\}  $ by solving problem \eqref{phi_convex}\;
$    \tilde{\psi}_{m,n}=   {\psi}_{m,n}, \forall m,n  $\;
}
Round each $ C_{m,n} $ by \eqref{roundC} based on $ \tilde{\psi}_{m,n} $\;
\KwOut{  $ \{C_{m,n}\}$ }
\end{algorithm}
\begin{algorithm}[h]
\caption{ Overall algorithm for solving problem \eqref{27}}\label{overall_alg}
\LinesNumbered
\Initialization{Set $ C_{prec}^{(0)},\{ p_{k,n}^{(0)} \} , \{ {C}_{m,n}^{(0)} \} $ to a feasible value and $ j = 0 $\;    }
\Repeat{ convergence}
{
$ j = j + 1 $\;
Solve sub-problem \eqref{obj_C_prec} by a line search method to obtain $ C_{\text{prec}}^{(j)}  $ with fixed $\{ p_{k,n}^{(j-1)}\} , \{{C}_{m,n}^{(j-1)}\}   $\;
Solve sub-problem \eqref{obj_p_all} by transforming the original problem into a FP one and using Algorithm \ref{algFP} to obtain $ \{ p_{k,n}^{(j)}\} $ with fixed $  C_{\text{prec}}^{(j)}, \{C_{m,n}^{(j-1)}\}$\;
Solve sub-problem \eqref{phi_all} by using Algorithm \ref{alg_phi} to obtain $ \{ C_{m,n}^{(j)}\} $ with fixed $C_{\text{prec}}^{(j)}, \{p_{k,n}^{(j)}\}  $\;
}
\KwOut{  $ \{C_{m,n}^{(j)}\}, \{p_{k,n}^{(j)}\}, C_{\text{prec}}^{(j)}$. }
\end{algorithm}

Here, we aim to obtain local optimal solution $\{ \psi _{m,n}\}$ with fixed $ \{p_{k,n} \} $.
First, we omit the integer requirement in constraint \eqref{inte_con} and adopt an alternating optimization based algorithm to solve
\begin{subequations}
\label{phi_all}
\begin{align}
\underset{\{\psi_{m,n}\} }{\operatorname{min}}\quad    {  E^{\text{trans}}_{\mathcal{K}}   }  \quad \text { s.t. }  \sum_{n=1}^{N} \frac{1}{2} \log_2(\psi_{m,n})\leq \bar{C}_m,\forall m.
\end{align}
\end{subequations}
%is defined in \eqref{phi_first}. $ E^{\text{trans}} _{\mathcal{K}} $ 
The objective function $ E^{\text{trans}} _{\mathcal{K}} $ in \eqref{phi_all} can be considered as the sum of two parts, i.e.,
\begin{equation}\label{sum_e}
\begin{aligned}
E^{\text{trans}}_{\mathcal{K}} = E^{\text{device}} + E^{\text{RRH}},
\end{aligned}
\end{equation}
where
\begin{equation}
\begin{aligned}\label{e_users}
&E^{\text{device}} = \sum_{k=1}^{K}   \frac{  \| \bm{d}^{k}_t   \| C_{\text{prec}} p_k /C_{\text{max}}  }{ \sum_{n\in\Omega_k}     \frac{B}{N} \log_2 \left(   1+ \sum_{m=1}^{M} \frac{ |   h_{m,k,n}  |^2p_{k,n} \psi_{m,n} }{  \sigma_{m,n}^2   \psi_{m,n} +        |   h_{m,k,n}  |^2  p_{k,n} + \sigma_{m,n}^2    }\right)   },
\end{aligned}
\end{equation}
and 
\begin{equation}
\begin{aligned}\label{e_rrhs}
&E^{\text{RRH}}=
& \sum_{m=1}^{M}\sum_{k=1}^{K}  {  \| \bm{d}^{k}_t   \| C_{\text{prec}}  P_{fl, m}  /C_{\text{max}}   \sum_{n\in\Omega_k}\log_2 \psi_{m,n}  }.
\end{aligned}
\end{equation}
$  E^{\text{device}} $ and $  E^{\text{RRH}} $ denote the energy consumed by devices and RRHs, respectively.
It can be shown that $ E^{\text{device}}$ is convex over $\{ \psi_{m,n}\} $ in \eqref{e_users}, $ \forall m, n $, while $E^{\text{RRH}}  $ is a non-convex function over $  \{\psi_{m,n} \}$ in \eqref{e_rrhs}. As a result, we apply the first-order approximation as an upper bound to $ E^{\text{RRH}} $, which is shown in \eqref{term2upper}.
\begin{equation}\label{term2upper}
\begin{aligned}
&E^{\text{RRH}} \leq { \sum_{m=1}^{M}\sum_{k=1}^{K}  {   \frac{\| \bm{d}^{k}_t   \| C_{\text{prec}}  P_{fl, m}   }{C_{\text{max}}  } \sum_{n\in\Omega_k}  \left(\log_2 (\tilde{\psi}_{m,n})   +  \frac{\psi_{m,n}- \tilde{\psi}_{m,n}}{  \tilde{\psi}_{m,n} \ln2}\right)}} \triangleq \hat{E}^{\text{RRH}} .
\end{aligned}
\end{equation}
Based on the upper bound in \eqref{term2upper}, we transform the non-convex objective function \eqref{sum_e} into a convex one. Similar approach can be applied to convexify the fronthaul link capacity constraints. Specifically, since $ C_m $ is concave over $ \psi_{m,n} $, its first-order approximation serves as an upper bound,
i.e.,
\begin{equation}
\begin{aligned}
C_m&= \sum_{n=1}^{N} \frac{1}{2}\log_2(\psi_{m,n}) \leq \sum_{n=1}^{N}\frac{1}{2} \left(\log_2 (\tilde{\psi}_{m,n})   +  \frac{\psi_{m,n}- \tilde{\psi}_{m,n}}{  \tilde{\psi}_{m,n} \ln2}\right),m=1,\cdots,M.
\end{aligned}
\end{equation}
Therefore, based on the above approximation, we obtain the following convex problem:
\begin{equation}
\label{phi_convex}
\begin{aligned}
\underset{\{\psi_{m,n}\} }{\operatorname{min}}   {  \hat{E}^{\text{trans}} _{\mathcal{K}} }  
\text { s.t. }   \sum_{n=1}^{N}\frac{1}{2} \left(\log_2 (\tilde{\psi}_{m,n})   +  \frac{\psi_{m,n}- \tilde{\psi}_{m,n}}{  \tilde{\psi}_{m,n} \ln2}\right) \leq \bar{C}_m,\forall m,
\end{aligned}
\end{equation}
where $\hat{E}^{\text{trans}} _{\mathcal{K}} \triangleq  E^{\text{device}} + \hat{E}^{\text{RRH}} $ is a convex function over $ \{\psi_{m,n}\} $, and $  E^{\text{user}}$ and $    \hat{E}^{\text{RRH}} $ are defined in \eqref{e_users} and \eqref{term2upper}, respectively. Thus, we can solve the problem using CVX toolbox \cite{cvx} \cite{gb08}.

%After obtaining solution $ \{\hat{p}_{k,n}, \hat{\psi}_{m,n}\} $, we fix $ p_{k,n} =\hat{p}_{k,n} $ and 
Next, we deal with the integer constraint in \eqref{inte_con}. Based on the approximate solution $ \tilde{\psi}_{m,n} $, we aims to find a feasible solution to $\psi_{m,n} $ such that $ C_{m,n} = \frac{1}{2}\log_2 \psi_{m,n} $ is an integer. In the following, for any given $ m $, we round each $  C_{m,n}  ,\forall m,n$ to its nearby integer as follows \cite{DBLP:journals/tcom/LiuBZ15}:

\begin{equation}\label{roundC}
\begin{aligned}
&\frac{1}{2}\log_2\psi_{{m},n} =\left\{
\begin{aligned}
\left\lfloor \frac{1}{2}\log_2 \tilde{\psi}_{{m},n}\right\rfloor & ,  \text{if }   \frac{1}{2}\log_2 \tilde{\psi}_{{m},n}-\left\lfloor \frac{1}{2}\log_2 \tilde{\psi}_{{m},n}\right\rfloor \leq \beta_m, \\
\left\lceil  \frac{1}{2}\log_2 \tilde{\psi}_{{m},n} \right\rceil  & , \text{otherwise},
\end{aligned}
\right.
\end{aligned}
\end{equation}
where $ 0 \leq \beta_m \leq 1, \forall m$. By searching $ {m} $ from $ 1 $ to $ M $, the overall feasible solution $ \{ C_{m,n}\} $ is obtained. We present the iterative method to solve problem \eqref{phi_all} in Algorithm \ref{alg_phi}. The overall iterative approach for solving problem \eqref{27} is given in Algorithm \ref{overall_alg}.

\section{Simulation Results }\label{results}
%In this section, we evaluate the performance of the proposed schemes to demonstrate the advantage of our proposed optimization approach. 
In this section, we conduct experiments to evaluate the performance of the proposed schemes and demonstrate the advantages of our proposed optimization approach.
Specifically, we uniformly deploy devices in a communication system of $ K=16 $ devices and $ M =5$ RRHs randomly distributed in a circle area of radius $ R=500 $ m. We set the bandwidth of wireless channel to  $ B=300  $ MHz where the bandwidth is equally divided into $ N=64 $ SCs and each user is pre-allocated $ N/K =4$ SCs. Note that we assume that the capacities of all fronthaul links are identical, i.e., $ \bar{G}_m = \bar{G} , \forall m$. Rayleigh fading is considered for all channels and the pass loss model is formulated as 
$L(d) = T_0\left(\frac{d}{d_0}\right)^{-\alpha}$,
where $ T_0 $ denotes the reference pathloss corresponding to $ d_0=1 $ m, $ \alpha $ denotes the pathloss exponent and $ d $ denotes the distance of the wireless link. In the following experiments, we set $ T_0 = 30 $ dB and $ \alpha = 3 $ for the wireless link between RRHs and devices. The maximum transmit power at each device is set to be $ 23 $ dBm, i.e., $ \bar{P}_k = \bar{P}, \forall k \in\mathcal{K} $. For other parameters, we set $ \bar{K} = 10 $, $ C_\text{max} = 32 $ bits, $ \sigma_k=1 $, $ I= 5 $, $ \beta=2/\mu $, $\epsilon=0.01  $, $ L=1 $, $ \mu= 0.89 $, $ \gamma=1 $, $ G=0.02 , \forall k \in \mathcal{K}, \beta_m = 0.5, \forall m \in \mathcal{M}$. In terms of the computing model, devices are assumed to have the same architecture of the processing chip, and we set $ A = 3.7 $ pJ, $ \alpha=1.25 $ and $ u =  64$ \cite{7870353}. A QNN structure with two convolutional layers is considered: 32 kernels of size $ 2\times2 $ with one padding and two of strides, and kernels of size $ 2\times 2 $ with one padding and two of strides, each followed by $ 2 \times 2 $ max pooling, two dense layers of $  2000$ neurons, and one fully-connected output layer of $ 10 $ neurons with local batch size $ 1 $. Thus, we have $ N_s= 0.28\times 10^6 $, $  N_c = 0.37\times 10^6$ and $ O_s = 2266 $. Note that $ d $ equals to the number of parameters $ N_s $. 
%The dataset we use is MNIST dataset. 
All statistic results are averaged over $ 1000 $ independent runs. 
\begin{figure}[htb]
\centering
\centerline{\includegraphics[scale=0.4]{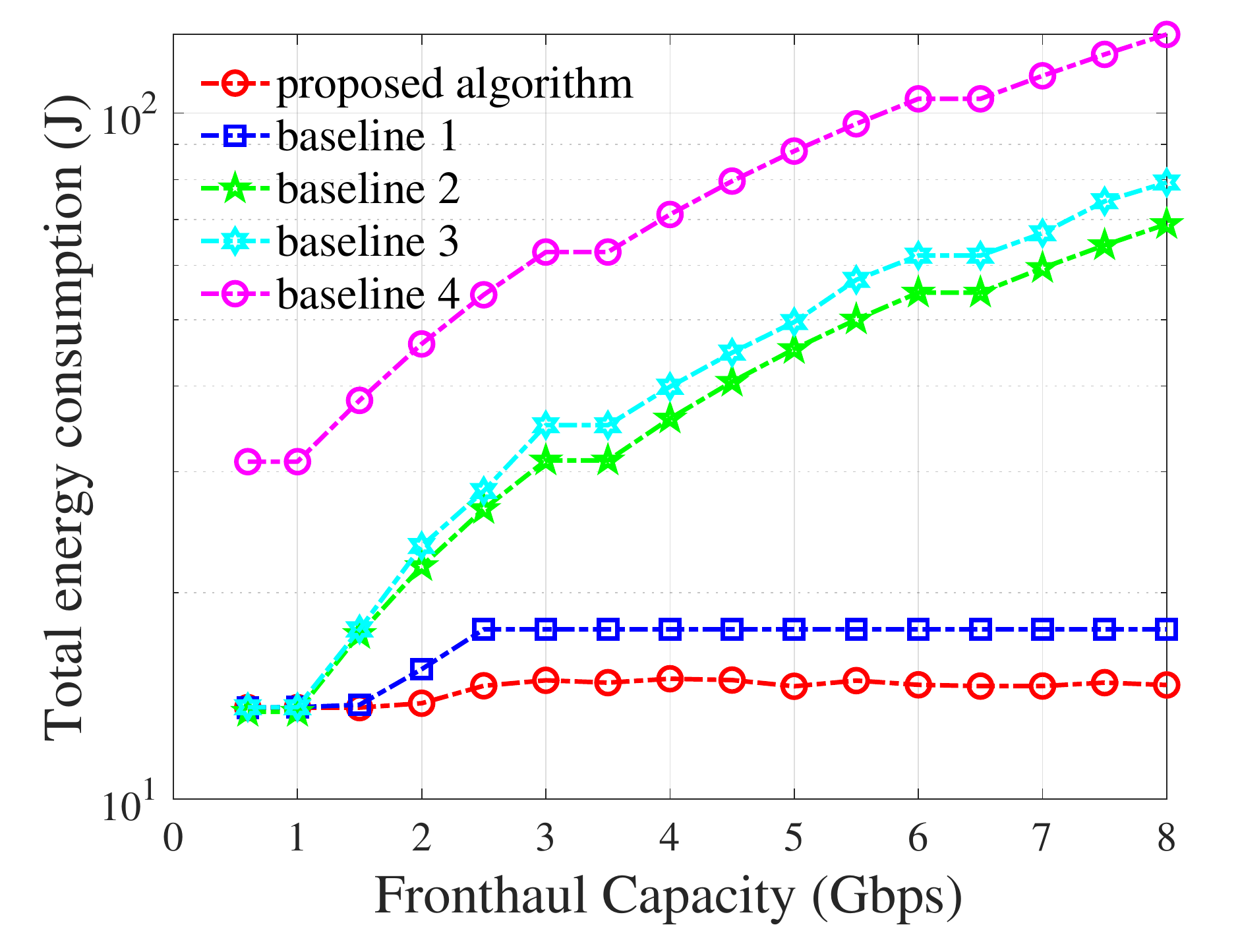}}
\caption{Energy power consumption due to joint optimization of user power allocation, fronthaul rate allocation and level of precision of QNN.}
\label{fig:methods}
\end{figure}

\begin{figure*}[htbp!]
\centering
\begin{minipage}{0.32\linewidth}
\centering
\includegraphics[width=\linewidth]{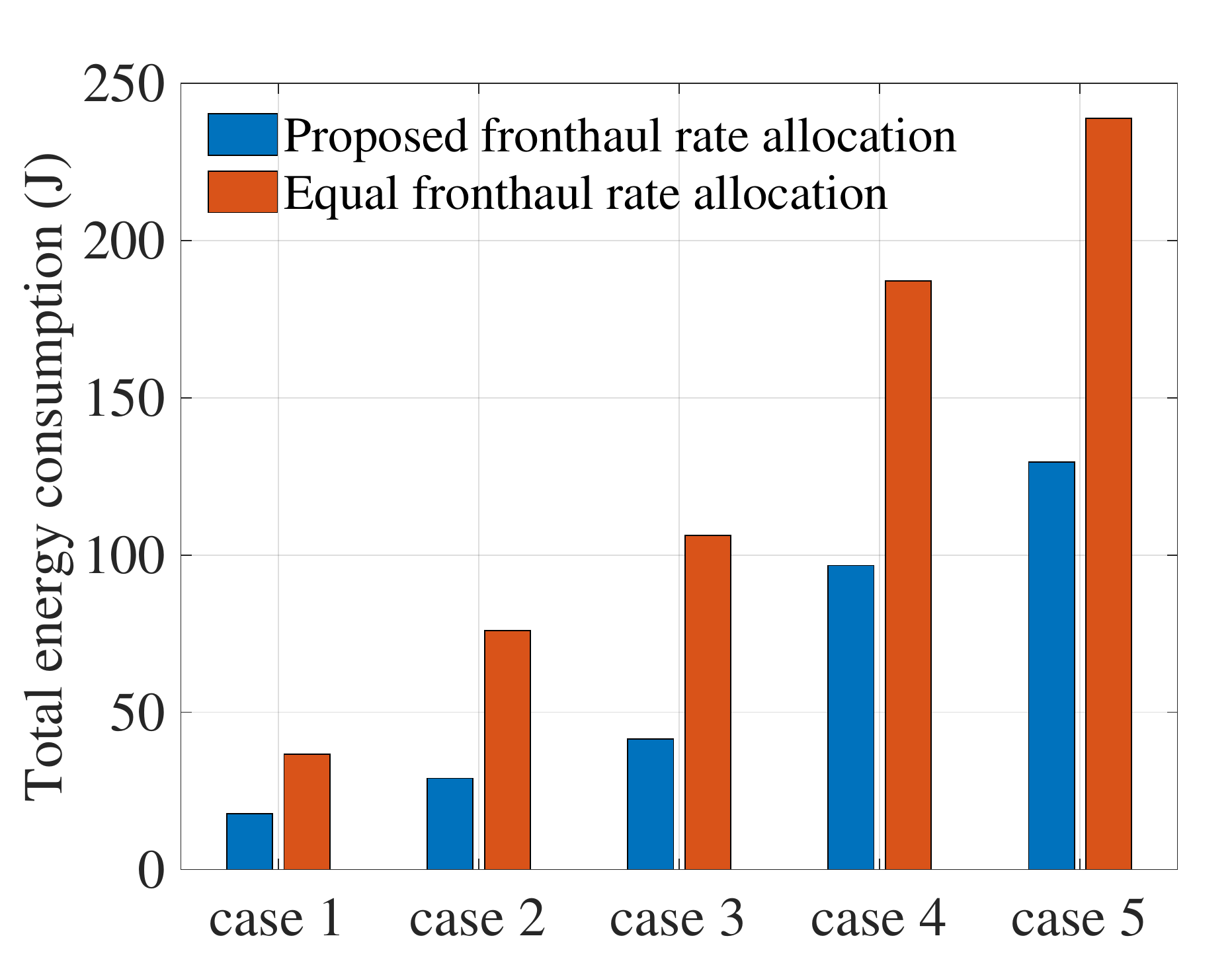}
\caption{Transmit energy consumption for various user power budgets.}
\label{fig:userpower}
\end{minipage}
%\qquad
\begin{minipage}{0.32\linewidth}
\centering
\includegraphics[width=\linewidth]{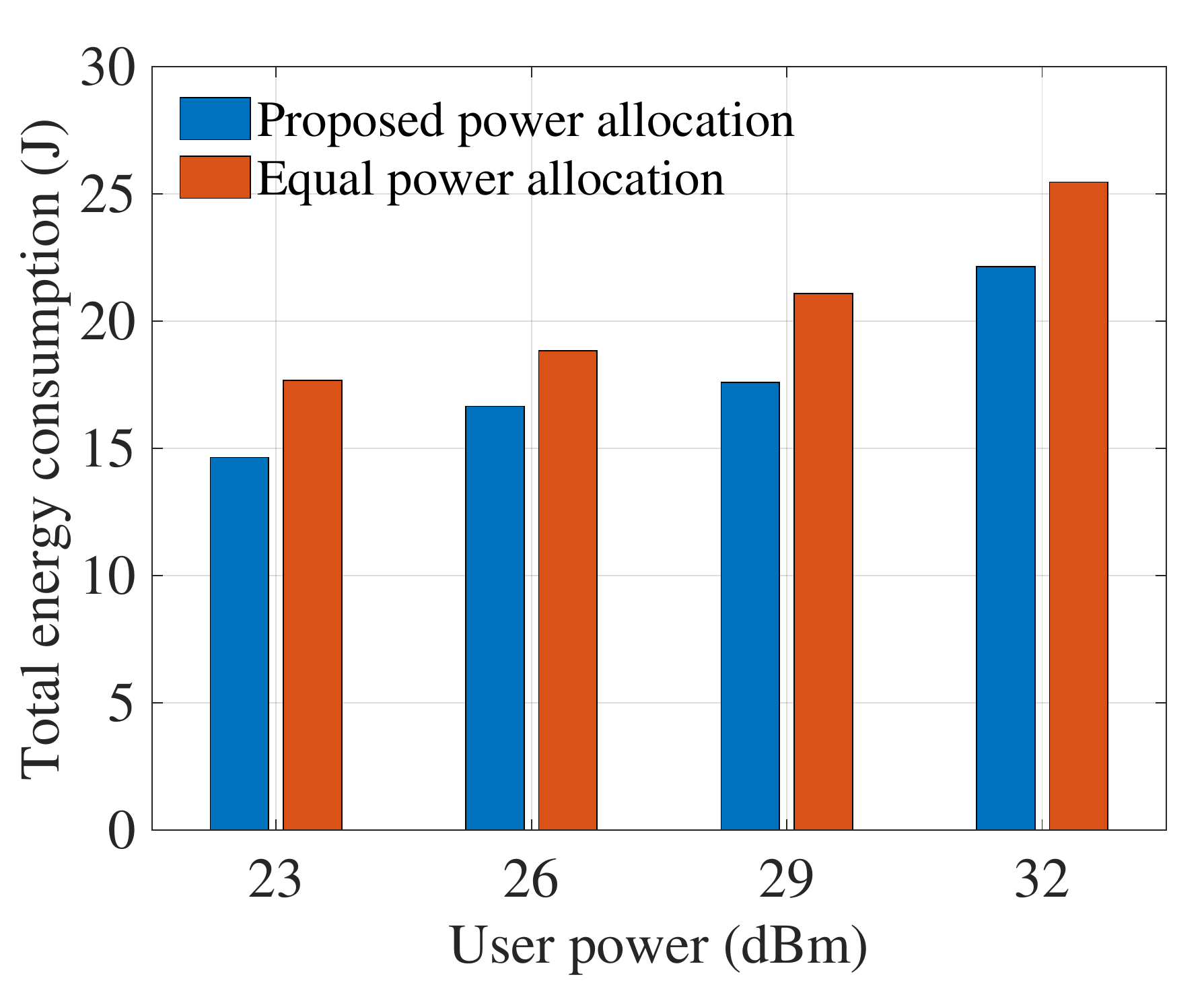}
\caption{Transmit energy consumption for various fronthual rate budgets.}
\label{fig:c}
\end{minipage}
\begin{minipage}{0.32\linewidth}
\centering
\includegraphics[width=\linewidth]{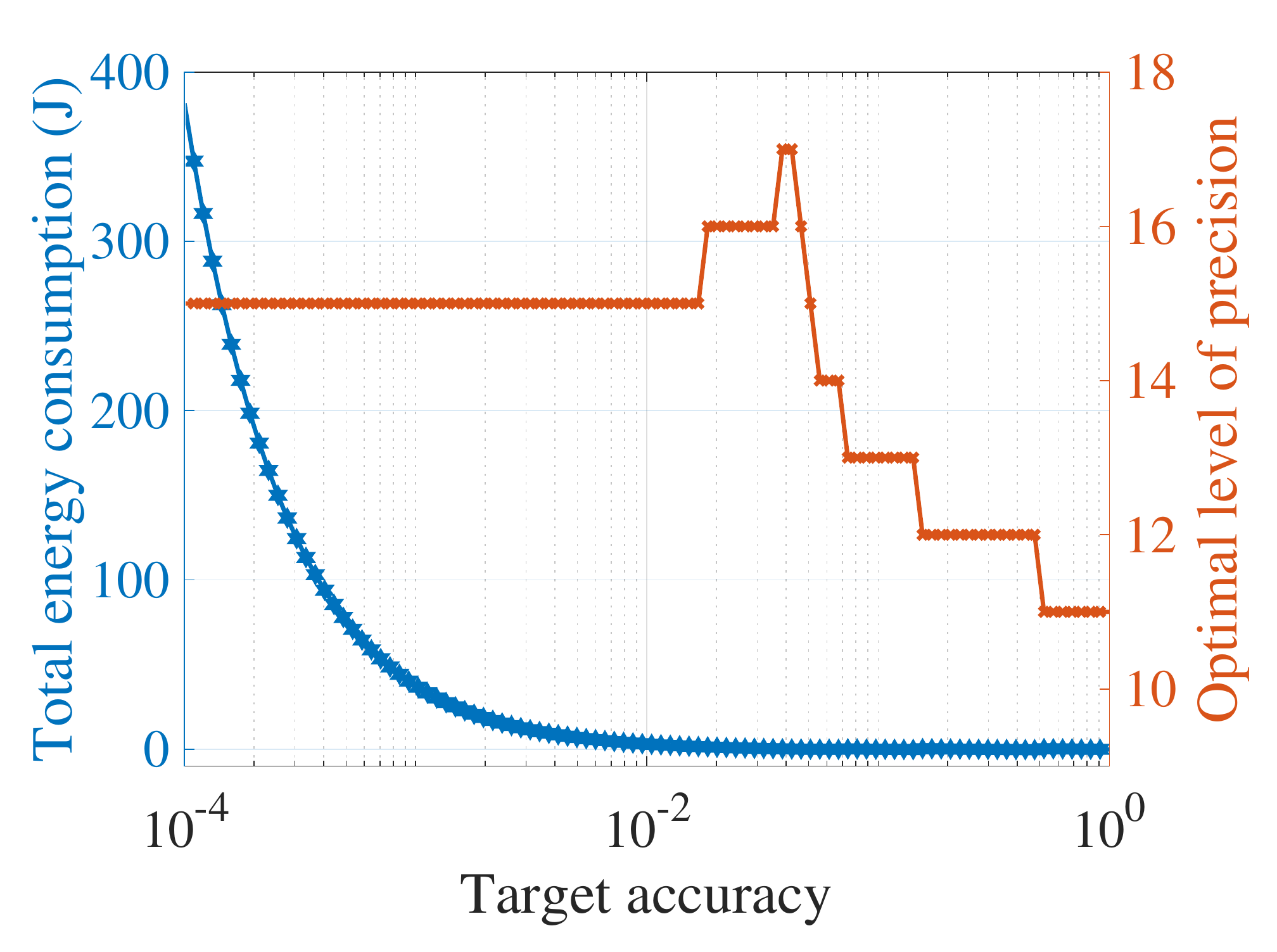}
\caption{Transmit energy consumption for different target accuracies.}
\label{fig:acc}
\end{minipage}
\end{figure*}
To illustrate the gain from joint optimization of wireless power, fronthaul rate allocation and level of precision of QNN, we consider the following four baseline schemes: 
%in the following we provide various benchmark schemes where transmit power, fronthaul rate allocation or level of precision of QNN are partially optimized.
\begin{itemize}
\item \textbf{ Baseline 1: Equal power with optimized fronthaul rate and level of precision.} In this scheme, the user allocates its transmit power equally among SCs, i,e., $ p_{n} = \bar{P }/ N$. Then, with equal power allocation, we optimize the fronthaul rate allocation at the RRHs and the level of precision of QNN \cite{DBLP:journals/corr/abs-2111-07911} so as to minimize the total energy consumption.
\item \textbf{ Baseline 2: Equal fronthaul rate with optimized power and level of precision.} In this scheme, the RRH equally allocates its fronthaul rate among SCs, $ G_{m, n} =  \bar{G} / N $. Then, with equal fronthaul rate allocation, we optimize the transmit power of the devices and the precision level of QNN \cite{DBLP:journals/corr/abs-2111-07911} in order to minimize the total energy consumption.
\item \textbf{ Baseline 3: Equal power and fronthaul rate with optimized level of precision.} In this scheme, the user allocates its transmit power equally among SCs, and each RRH equally allocates its fronthaul link bandwidth among SCs. Based on the given equal power allocation and fronthaul rate allocation, we optimize precision level of QNN \cite{DBLP:journals/corr/abs-2111-07911} to minimize the total energy consumption.
\item \textbf{ Baseline 4: Equal power and fronthaul rate with level of precision 31 bit.} In this scheme, the user allocates its transmit power equally among SCs, each RRH equally allocates its fronthaul link bandwidth among SCs, and the level of precision of QNN is set to be $ 31 $ bits.
\end{itemize}

In Fig. \ref{fig:methods}, we evaluate  the performance of the proposed optimization scheme against that of other baseline schemes. Fig. \ref{fig:methods} shows the total consumption of the FL system until convergence for varying optimization schemes versus total fronthaul capacity $  \bar{G} $. It is observed that compared with baseline 1-4 where only either wireless power, fronthaul rate allocation or level of precision is optimized, the joint optimization algorithm proposed in Section \ref{solution} achieves a much lower energy consumption, especially, when the fronthaul link capacity is small. Furthermore, it is observed from baseline 1 and the proposed algorithm that when the fronthaul link capacity $  \bar{G} $ is large, power allocation optimization plays the dominant role in improving the total energy consumption, while when $  \bar{G} $ is small, the performance of the two methods is similar. We can also observe the similar tendency from baseline 2 and 3. Furthermore, when $  \bar{G} $ is sufficiently large, fronthaul rate optimization has a stronger impact to the performance of energy consumption than power allocation optimization. Since the influence of user transmit power on the total energy consumption becomes weak when $  \bar{G} $ is sufficiently large. The scheme with fronthaul rate optimization can achieve a better performance. Fig. \ref{fig:methods} also presents that schemes with the precision level optimization outperforms baseline 4. This indicates that the level of precision optimization approach effectively improves the energy consumption of the whole system.

% \begin{figure}[htb]
% \centering
% \centerline{\includegraphics[scale=0.42]{fig/power_energy_all.eps}}
% \caption{Transmit energy consumption for various user power budgets.}
% \label{fig:userpower}
% \end{figure}

In Fig. \ref{fig:userpower}, we investigate the impact of user power constraint on the performance of the energy consumption. The maximum user transmit power budget is set to be $ 23 , 25, 29, 32$ dBm. Fig. \ref{fig:userpower} compares the total energy consumption until convergence for varying the user power budget $ \bar{P} $. We observe that the proposed power allocation optimization outperforms equal power allocation. In particular, for the case $ \bar{P} = 23 $ dBm, we can reduce the total energy consumption up to 20\% compared to the equal power allocation scheme due to the effectiveness of the proposed transmit power optimization algorithm.

% \begin{figure}[htb]
% \centering
% \centerline{\includegraphics[scale=0.42]{fig/c_energy_model_all.eps}}
% \caption{Transmit energy consumption for various fronthual rate budgets.}
% \label{fig:c}
% \end{figure}

In Fig. \ref{fig:c}, we compare the total energy consumption until convergence for different sizes of QNN models. The batchsize is set to be $ 16 $ for Case 1-5. Case 1 considers the same model structure as mentioned before and we have $ N_s=0.28\times 10^6 , N_c = 11.85\times 10^6$ and $ O_s = 2266 $ with batchsize $ =16 $. 
Case 2 considers a three-layer CNN, where the first two convolutional layer settings are the same as Case $ 1 $, the third convolutional layer is set to be $ 32 $ kernels of size $ 2\times 2 $ with one padding and one stride, and the dense layer and the fully-connected output layer are set to be the same as Case 1. For Case 2, we have $N_s = 1.05\times  10^6   $, $ N_c = 39.70 \times 10^6 $ and $ O_s = 2394 $. 
For Case 3, we add an additional dense layer with $ 1000 $ neurons where $N_s = 2.08\times  10^6   $, $ N_c = 70.56 \times 10^6 $ and $ O_s = 3394 $. 
For Case $ 4 $, we add an additional dense layer with $ 2000 $ neurons to Case $ 3 $, where $N_s = 6.08\times  10^6   $, $ N_c = 198.56 \times 10^6 $ and $ O_s = 5368 $. 
For Case $ 5$, we add an additional dense layer with $ 1000 $ neurons to Case $ 4 $, where $ N_s = 7.08\times 10^6  $ and $ N_c = 230.56 \times 10^6 $ and $ O_s =  6368$.
It can be observed that our proposed joint optimization scheme is effective for CNN models with a large model size. It is capable of decreasing the total energy consumption up to 50\% compared to the equal fronthaul rate allocation for Case 1-5.
% \begin{figure}[htb]
% \centering
% \centerline{\includegraphics[scale=0.42]{fig/energy_accuracy.eps}}
% \caption{Transmit energy consumption for different target accuracies.}
% \label{fig:acc}
% \end{figure}

In Fig. \ref{fig:acc}, we illustrate the relationship between the total energy consumption and the optimal level of precision via varying the target accuracy. The same CNN model as mentioned in Fig. \ref{fig:methods} is used. It can be observed that the optimal level of precision increases and then decreases with target accuracy. There is a tradeoff between FL performance with target accuracy and the optimal precision level. To achieve high model performance, i.e., lower target accuracy $ 10^{-4} $, a relatively high level of precision is needed to maintain more model updates information and mitigate the quantization error. Meanwhile, as the target accuracy becomes looser, i.e., $4\times 10^{-2}  $, the optimal level of precision increases because with higher level of precision, i.e., $ 17 $ bits, the global communication rounds can be decreased and much lower communication energy consumption can be achieved. Lower total energy consumption can be obtained with a relatively high level of precision. Moreover, as the target accuracy increases, a lower level of precision can be chosen and a fewer number of global iterations are needed to achieve the target accuracy. Hence, an optimal level of precision can be effectively achieved by our proposed joint optimization algorithm. 
\begin{figure}[htb]
\centering
\centerline{\includegraphics[scale=0.4]{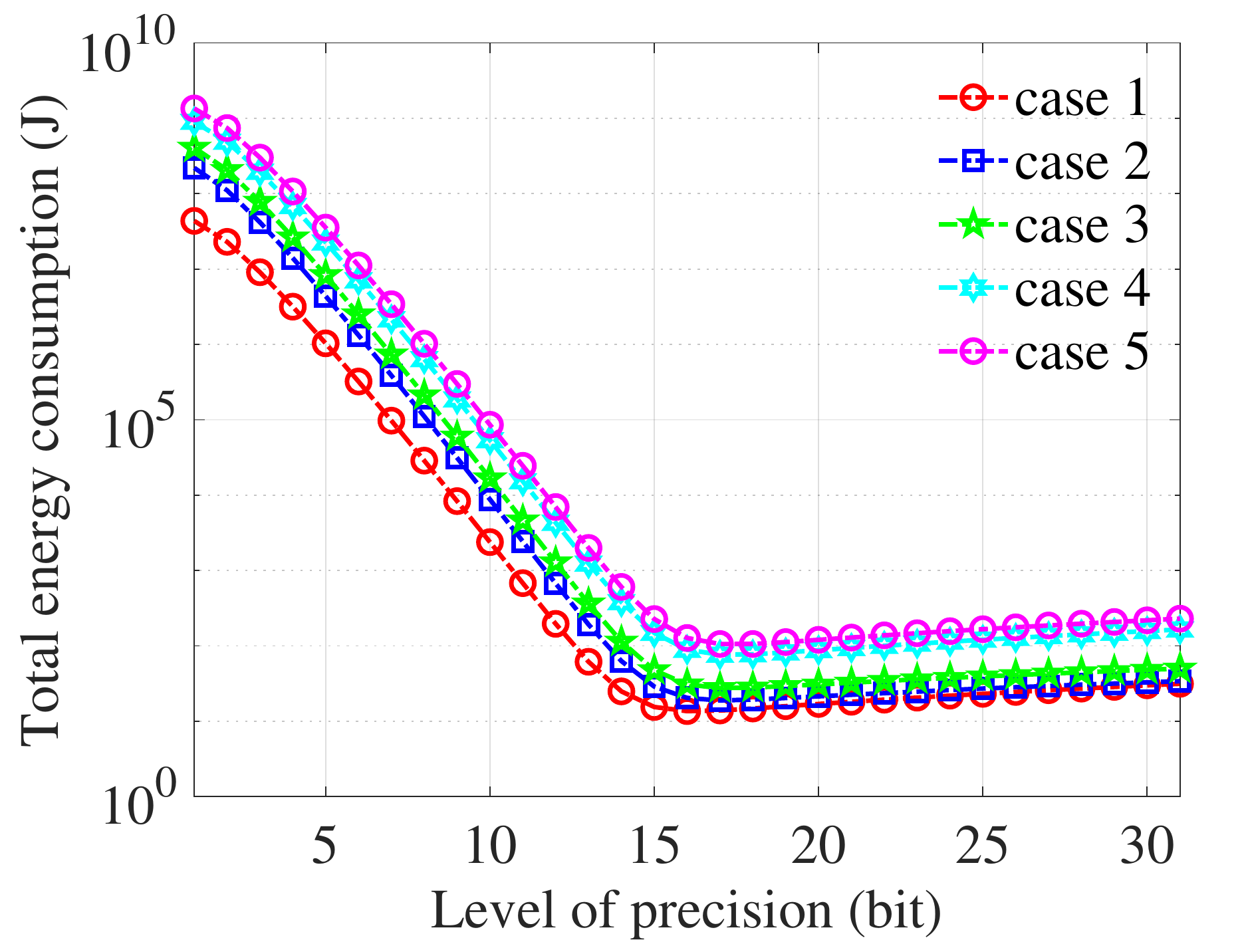}}
\caption{Transmit energy consumption for different level of precisions.}
\label{fig:c_prec_en}
\end{figure}
In Fig. \ref{fig:c_prec_en}, we illustrate the total energy consumption versus level of precision for the aforementioned five schemes. The model size is the same as in Fig. \ref{fig:c}. We can observe that as the model size increases, the total energy consumption also increases. Furthermore, when the level of precision becomes larger, the total energy consumption decreases promptly and then increases slowly. Thus, there is a tradeoff between level of precision and the total energy consumption. A very small or very high level of precision may induce an undesired large quantization error or unnecessary energy consumption. We observe that the corresponding optimal level of precision for Case 1-5 are $15,16,16,17,17  $ as illustrated in Fig. \ref{fig:c_prec_en}. 

Based on the aforementioned simulations, we conclude that our proposed framework requires much less energy than a standard FL model. We discover the tradeoff exists between the energy consumption and the convergence accuracy with respect to user transmit power, fronthaul rate, as well as the precision level. In particular, a very small level of precision may induce an undesired quantization error leading to much more rounds of iterations, while a very large energy level of precision may also bring extra energy consumption. The proposed framework demonstrates significant performance benefits thanks to the joint design of the precision level, the fronthaul rate, and user transmit power. 

\section{conclusion} \label{conclusion} 
This work presented an effective and practical way to implement FL on energy-constrained devices, by introducing QNNs to simplify the computation complexity, as well as OFDMA-based Cloud-RAN to improve the communication performance. To this end, a rigorous energy consumption model for computing energy consumption of QNNs and communication energy consumption over Cloud-RAN was introduced. Based on the rigorous convergence analysis, we then formulated and solved an energy minimization problem to jointly optimize the power allocation, fronthaul rate allocation, and precision levels of QNNs subject to a target accuracy and limited fronthaul link capacity. Simulation results have shown that our proposed algorithm can significantly reduce the energy consumption while balancing the tradeoff between energy efficiency and convergence rate.
\bibliographystyle{IEEEtran}
\bibliography{main} 

\begin{thebibliography}{10}
\providecommand{\url}[1]{#1}
\csname url@samestyle\endcsname
\providecommand{\newblock}{\relax}
\providecommand{\bibinfo}[2]{#2}
\providecommand{\BIBentrySTDinterwordspacing}{\spaceskip=0pt\relax}
\providecommand{\BIBentryALTinterwordstretchfactor}{4}
\providecommand{\BIBentryALTinterwordspacing}{\spaceskip=\fontdimen2\font plus
\BIBentryALTinterwordstretchfactor\fontdimen3\font minus
  \fontdimen4\font\relax}
\providecommand{\BIBforeignlanguage}[2]{{%
\expandafter\ifx\csname l@#1\endcsname\relax
\typeout{** WARNING: IEEEtran.bst: No hyphenation pattern has been}%
\typeout{** loaded for the language `#1'. Using the pattern for}%
\typeout{** the default language instead.}%
\else
\language=\csname l@#1\endcsname
\fi
#2}}
\providecommand{\BIBdecl}{\relax}
\BIBdecl

\bibitem{10.1145/3298981}
Q.~Yang, Y.~Liu, T.~Chen, and Y.~Tong, ``Federated machine learning: Concept
  and applications,'' \emph{ACM Transactions on Intelligent Systems and
  Technology (TIST)}, vol.~10, no.~2, pp. 1--19, 2019.

\bibitem{9606720}
K.~B. Letaief, Y.~Shi, J.~Lu, and J.~Lu, ``Edge artificial intelligence for 6g:
  Vision, enabling technologies, and applications,'' \emph{IEEE J. Sel. Areas
  Commun.}, vol.~40, no.~1, pp. 5--36, 2022.

\bibitem{DBLP:journals/cm/NiknamDR20}
S.~Niknam, H.~S. Dhillon, and J.~H. Reed, ``Federated learning for wireless
  communications: Motivation, opportunities, and challenges,'' \emph{{IEEE}
  Commun. Mag.}, vol.~58, no.~6, pp. 46--51, 2020.

\bibitem{DBLP:journals/corr/KonecnyMR15}
J.~Kone{\v{c}}n{\'y}, B.~McMahan, and D.~Ramage, ``Federated optimization:
  Distributed optimization beyond the datacenter,'' \emph{arXiv preprint
  arXiv:1511.03575}, 2015.

\bibitem{shi2023taskoriented}
Y.~Shi, Y.~Zhou, D.~Wen, Y.~Wu, C.~Jiang, and K.~B. Letaief, ``Task-oriented
  communications for 6g: Vision, principles, and technologies,'' \emph{arXiv
  preprint arXiv:2303.10920}, 2023.

\bibitem{9502547}
Z.~Wang, J.~Qiu, Y.~Zhou, Y.~Shi, L.~Fu, W.~Chen, and K.~B. Letaief,
  ``Federated learning via intelligent reflecting surface,'' \emph{IEEE Trans
  on Wireless Commun.}, vol.~21, no.~2, pp. 808--822, 2022.

\bibitem{9134426}
Y.~Shi, K.~Yang, T.~Jiang, J.~Zhang, and K.~B. Letaief,
  ``Communication-efficient edge ai: Algorithms and systems,'' \emph{IEEE
  Commun. Surveys Tuts.}, vol.~22, no.~4, pp. 2167--2191, 2020.

\bibitem{he2023reconfigurable}
J.~He, Y.~Mao, Y.~Zhou, T.~Wang, and Y.~Shi, ``Reconfigurable intelligent
  surfaces empowered green wireless networks with user admission control,''
  \emph{arXiv preprint arXiv:2206.07987}, 2023.

\bibitem{9810113}
P.~Yang, Y.~Jiang, T.~Wang, Y.~Zhou, Y.~Shi, and C.~N. Jones, ``Over-the-air
  federated learning via second-order optimization,'' \emph{IEEE Trans. on
  Wireless Commun.}, vol.~21, no.~12, pp. 10\,560--10\,575, 2022.

\bibitem{9042352}
M.~Mohammadi~Amiri and D.~Gündüz, ``Machine learning at the wireless edge:
  Distributed stochastic gradient descent over-the-air,'' \emph{IEEE Trans. on
  Signal Processing}, vol.~68, pp. 2155--2169, 2020.

\bibitem{9014530}
M.~M. Amiri and D.~Gündüz, ``Federated learning over wireless fading
  channels,'' \emph{IEEE Trans. on Wireless Commun.}, vol.~19, no.~5, pp.
  3546--3557, 2020.

\bibitem{9912341}
X.~Fan, Y.~Wang, Y.~Huo, and Z.~Tian, ``1-bit compressive sensing for efficient
  federated learning over the air,'' \emph{IEEE Trans. on Wireless Commun.},
  pp. 1--1, 2022.

\bibitem{9562516}
D.~Fan, X.~Yuan, and Y.-J.~A. Zhang, ``Temporal-structure-assisted gradient
  aggregation for over-the-air federated edge learning,'' \emph{IEEE J. Sel.
  Areas Commun.}, vol.~39, no.~12, pp. 3757--3771, 2021.

\bibitem{9269459}
Y.-S. Jeon, M.~M. Amiri, J.~Li, and H.~V. Poor, ``A compressive sensing
  approach for federated learning over massive mimo communication systems,''
  \emph{IEEE Trans. on Wireless Commun.}, vol.~20, no.~3, pp. 1990--2004, 2021.

\bibitem{9878377}
Y.~Oh, N.~Lee, Y.-S. Jeon, and H.~Vincent~Poor, ``Communication-efficient
  federated learning via quantized compressed sensing,'' \emph{IEEE Trans. on
  Wireless Commun.}, pp. 1--1, 2022.

\bibitem{10.1007/3-540-45465-9_59}
M.~Charikar, K.~Chen, and M.~Farach-Colton, ``Finding frequent items in data
  streams,'' in \emph{Automata, Languages and Programming}, P.~Widmayer,
  S.~Eidenbenz, F.~Triguero, R.~Morales, R.~Conejo, and M.~Hennessy, Eds.\hskip
  1em plus 0.5em minus 0.4em\relax Springer Berlin Heidelberg, 2002, pp.
  693--703.

\bibitem{pmlr-v97-spring19a}
R.~Spring, A.~Kyrillidis, V.~Mohan, and A.~Shrivastava, ``Compressing gradient
  optimizers via count-sketches,'' in \emph{Proc. of International Conference
  on Machine Learning}, vol.~97.\hskip 1em plus 0.5em minus 0.4em\relax PMLR,
  2019, pp. 5946--5955.

\bibitem{pmlr-v119-rothchild20a}
D.~Rothchild, A.~Panda, E.~Ullah, N.~Ivkin, I.~Stoica, V.~Braverman,
  J.~Gonzalez, and R.~Arora, ``{F}etch{SGD}: Communication-efficient federated
  learning with sketching,'' in \emph{Proc. of International Conference on
  Machine Learning}, vol. 119.\hskip 1em plus 0.5em minus 0.4em\relax PMLR,
  2020, pp. 8253--8265.

\bibitem{9807354}
S.~Savazzi, V.~Rampa, S.~Kianoush, and M.~Bennis, ``An energy and carbon
  footprint analysis of distributed and federated learning,'' \emph{{IEEE}
  Trans. Green Commun. Netw.}, pp. 1--1, 2022.

\bibitem{9673130}
R.~Jin, X.~He, and H.~Dai, ``Communication efficient federated learning with
  energy awareness over wireless networks,'' \emph{IEEE Trans. on Wireless
  Commun.}, vol.~21, no.~7, pp. 5204--5219, 2022.

\bibitem{9264742}
Z.~Yang, M.~Chen, W.~Saad, C.~S. Hong, and M.~Shikh-Bahaei, ``Energy efficient
  federated learning over wireless communication networks,'' \emph{IEEE Trans.
  on Wireless Commun.}, vol.~20, no.~3, pp. 1935--1949, 2021.

\bibitem{9953187}
T.~T. Vu, H.~Q. Ngo, M.~N. Dao, D.~T. Ngo, E.~G. Larsson, and T.~Le-Ngoc,
  ``Energy-efficient massive mimo for federated learning: Transmission designs
  and resource allocations,'' \emph{IEEE Open J. of the Commun. Soc.}, vol.~3,
  pp. 2329--2346, 2022.

\bibitem{10001623}
L.~Lei, Y.~Yuan, Y.~Yang, Y.~Luo, L.~Pu, and S.~Chatzinotas, ``Sparsification
  and optimization for energy-efficient federated learning in wireless edge
  networks,'' in \emph{GLOBECOM 2022 - 2022 IEEE Global Communications
  Conference}, 2022, pp. 3071--3076.

\bibitem{YIN2023}
L.~Yin, S.~Lin, Z.~Sun, R.~Li, Y.~He, and Z.~Hao, ``A game-theoretic approach
  for federated learning: A trade-off among privacy, accuracy and energy,''
  \emph{Digital Communications and Networks}, 2023.

\bibitem{DBLP:journals/monet/HuHY22}
Y.~Hu, H.~Huang, and N.~Yu, ``Resource optimization and device scheduling for
  flexible federated edge learning with tradeoff between energy consumption and
  model performance,'' \emph{Mob. Networks Appl.}, vol.~27, no.~5, pp.
  2118--2137, 2022.

\bibitem{DBLP:journals/corr/abs-2111-07911}
M.~Kim, W.~Saad, M.~Mozaffari, and M.~Debbah, ``On the tradeoff between energy,
  precision, and accuracy in federated quantized neural networks,'' in
  \emph{{IEEE} International Conference on Communications {ICC}}.\hskip 1em
  plus 0.5em minus 0.4em\relax {IEEE}, 2022, pp. 2194--2199.

\bibitem{DBLP:journals/corr/abs-2207-09387}
------, ``Green, quantized federated learning over wireless networks: An
  energy-efficient design,'' \emph{arXiv preprint arXiv:2207.09387}, 2022.

\bibitem{6786060}
Y.~Shi, J.~Zhang, and K.~B. Letaief, ``Group sparse beamforming for green
  cloud-ran,'' \emph{IEEE Transactions on Wireless Communications}, vol.~13,
  no.~5, pp. 2809--2823, 2014.

\bibitem{DBLP:journals/tcom/LiuBZ15}
L.~Liu, S.~Bi, and R.~Zhang, ``Joint power control and fronthaul rate
  allocation for throughput maximization in ofdma-based cloud radio access
  network,'' \emph{{IEEE} Trans. Commun.}, vol.~63, no.~11, pp. 4097--4110,
  2015.

\bibitem{stephen2017joint}
R.~G. Stephen and R.~Zhang, ``Joint millimeter-wave fronthaul and ofdma
  resource allocation in ultra-dense cran,'' \emph{{IEEE} Trans. Commun.},
  vol.~65, no.~3, pp. 1411--1423, 2017.

\bibitem{mcmahan2017communication}
B.~McMahan, E.~Moore, D.~Ramage, S.~Hampson, and B.~A. y~Arcas,
  ``Communication-efficient learning of deep networks from decentralized
  data,'' in \emph{Artificial intelligence and statistics}.\hskip 1em plus
  0.5em minus 0.4em\relax PMLR, 2017, pp. 1273--1282.

\bibitem{DBLP:journals/jsac/ZhengSC21}
S.~Zheng, C.~Shen, and X.~Chen, ``Design and analysis of uplink and downlink
  communications for federated learning,'' \emph{{IEEE} J. Sel. Areas Commun.},
  vol.~39, no.~7, pp. 2150--2167, 2021.

\bibitem{8952884}
K.~Yang, T.~Jiang, Y.~Shi, and Z.~Ding, ``Federated learning via over-the-air
  computation,'' \emph{IEEE Trans. on Wireless Commun.}, vol.~19, no.~3, pp.
  2022--2035, 2020.

\bibitem{DBLP:journals/jsac/WangXSC22}
Y.~Wang, Y.~Xu, Q.~Shi, and T.~Chang, ``Quantized federated learning under
  transmission delay and outage constraints,'' \emph{{IEEE} J. Sel. Areas
  Commun.}, vol.~40, no.~1, pp. 323--341, 2022.

\bibitem{DBLP:journals/corr/abs-2006-10672}
M.~M. Amiri, D.~G{\"{u}}nd{\"{u}}z, S.~R. Kulkarni, and H.~V. Poor, ``Federated
  learning with quantized global model updates,'' \emph{arXiv preprint
  arXiv:2006.10672}, 2020.

\bibitem{DBLP:conf/acssc/MoonsGBV17}
B.~Moons, K.~Goetschalckx, N.~V. Berckelaer, and M.~Verhelst, ``Minimum energy
  quantized neural networks,'' in \emph{Asilomar Conference on Signals,
  Systems, and Computers ({ACSSC})}, M.~B. Matthews, Ed.\hskip 1em plus 0.5em
  minus 0.4em\relax {IEEE}, 2017, pp. 1921--1925.

\bibitem{Moons2019}
B.~Moons, D.~Bankman, and M.~Verhelst, \emph{Embedded Deep Neural
  Networks}.\hskip 1em plus 0.5em minus 0.4em\relax Cham: Springer
  International Publishing, 2019, pp. 1--31.

\bibitem{DBLP:journals/jmlr/HubaraCSEB17}
I.~Hubara, M.~Courbariaux, D.~Soudry, R.~El{-}Yaniv, and Y.~Bengio, ``Quantized
  neural networks: Training neural networks with low precision weights and
  activations,'' \emph{J. Mach. Learn. Res.}, vol.~18, pp. 187:1--187:30, 2017.

\bibitem{DBLP:journals/corr/abs-2202-07775}
G.~Interdonato and S.~Buzzi, ``The promising marriage of mobile edge computing
  and cell-free massive {MIMO},'' in \emph{{IEEE} International Conference on
  Communications, {ICC}}, 2022, pp. 243--248.

\bibitem{DBLP:conf/icassp/LiLL18}
Q.~Li, J.~Lei, and J.~Lin, ``Min-max latency optimization for multiuser
  computation offloading in fog-radio access networks,'' in \emph{{IEEE}
  International Conference on Acoustics, Speech and Signal Processing,
  {ICASSP}}.\hskip 1em plus 0.5em minus 0.4em\relax {IEEE}, 2018, pp.
  3754--3758.

\bibitem{DBLP:journals/tgcn/NgoTDML18}
H.~Q. Ngo, L.~Tran, T.~Q. Duong, M.~Matthaiou, and E.~G. Larsson, ``On the
  total energy efficiency of cell-free massive {MIMO},'' \emph{{IEEE} Trans.
  Green Commun. Netw.}, vol.~2, no.~1, pp. 25--39, 2018.

\bibitem{DBLP:journals/twc/MasoumiE20}
H.~Masoumi and M.~J. Emadi, ``Performance analysis of cell-free massive {MIMO}
  system with limited fronthaul capacity and hardware impairments,''
  \emph{{IEEE} Trans. Wirel. Commun.}, vol.~19, no.~2, pp. 1038--1053, 2020.

\bibitem{pmlr-v54-mcmahan17a}
B.~McMahan, E.~Moore, D.~Ramage, S.~Hampson, and B.~A.~y. Arcas,
  ``{Communication-Efficient Learning of Deep Networks from Decentralized
  Data},'' in \emph{Proc. of International Conference on Artificial
  Intelligence and Statistics}, vol.~54.\hskip 1em plus 0.5em minus 0.4em\relax
  PMLR, 20--22 Apr 2017, pp. 1273--1282.

\bibitem{9488679}
B.~Luo, X.~Li, S.~Wang, J.~Huang, and L.~Tassiulas, ``Cost-effective federated
  learning design,'' in \emph{IEEE INFOCOM}, 2021, pp. 1--10.

\bibitem{DBLP:books/sp/BauschkeC11}
H.~H. Bauschke and P.~L. Combettes, \emph{Convex Analysis and Monotone Operator
  Theory in Hilbert Spaces}, ser. {CMS} Books in Mathematics.\hskip 1em plus
  0.5em minus 0.4em\relax Springer, 2011.

\bibitem{DBLP:books/sp/NocedalW99}
J.~Nocedal and S.~J. Wright, \emph{Numerical Optimization}.\hskip 1em plus
  0.5em minus 0.4em\relax Springer, 1999.

\bibitem{DBLP:journals/tsp/ShenY18}
K.~Shen and W.~Yu, ``Fractional programming for communication systems - part
  {I:} power control and beamforming,'' \emph{{IEEE} Trans. Signal Process.},
  vol.~66, no.~10, pp. 2616--2630, 2018.

\bibitem{cvx}
I.~CVX~Research, ``{CVX}: Matlab software for disciplined convex programming,
  version 2.0,'' \url{http://cvxr.com/cvx}, Aug. 2012.

\bibitem{gb08}
M.~Grant and S.~Boyd, ``Graph implementations for nonsmooth convex programs,''
  in \emph{Recent Advances in Learning and Control}.\hskip 1em plus 0.5em minus
  0.4em\relax Springer-Verlag Limited, 2008, pp. 95--110.

\bibitem{7870353}
B.~Moons, R.~Uytterhoeven, W.~Dehaene, and M.~Verhelst, ``14.5 envision: A
  0.26-to-10tops/w subword-parallel dynamic-voltage-accuracy-frequency-scalable
  convolutional neural network processor in 28nm fdsoi,'' in \emph{IEEE
  International Solid-State Circuits Conference (ISSCC)}, 2017, pp. 246--247.

\bibitem{DBLP:conf/iclr/LiHYWZ20}
X.~Li, K.~Huang, W.~Yang, S.~Wang, and Z.~Zhang, ``On the convergence of fedavg
  on non-iid data,'' in \emph{Proc. of International Conference on Learning
  Representations (ICLR)}, 2020.

\end{thebibliography}
\begin{appendices} 
\section{Proof of Lemma 1}\label{app1}

With a random vector $ \bm{w} $ and a quantization level $ C_{\text{prec}} $, we can obtain that the quantized $ Q(w) $ is unbiasedly estimated where $ w \in \bm{w} $. First, we have
\begin{equation}
\begin{aligned}
\expp{ Q(w)} &= \operatorname{sign}(w) \cdot s_{i} \cdot \frac{s_{i+1} - |w|}{s_{i+1}- s_i} + \operatorname{sign}(w) \cdot s_{i+1}  \cdot \frac{|w|- s_i}{s_{i+1}- s_i }\\
&= \frac{\operatorname{sign}(w) }{ s_{i+1}- s_i}   ( s_{i} \cdot (s_{i+1} - |w|) +   s_{i+1} \cdot (|w|- s_i)   )=\operatorname{sign}(w) |w| = w.
\end{aligned}
\end{equation}
As for the quantization error, we have
%&=(s_i - |w|)^2    \frac{s_{i+1} - |w|}{s_{i+1}- s_i} + (s_{i+1} - |w|)^2  \cdot   \frac{|w|- s_i}{s_{i+1}- s_i }\\
\begin{equation}
\begin{aligned}
\expp{\norm{ Q(w) - w}^2} &= (\operatorname{sign}(w) \cdot s_{i}  -w)^2 \cdot   \frac{s_{i+1} - |w|}{s_{i+1}- s_i} + (\operatorname{sign}(w) \cdot s_{i+1}-w )^2 \cdot   \frac{|w|- s_i}{s_{i+1}- s_i }\\
&=( |w| - s_i   )(  s_{i+1} - |w|) \leq \left(\frac{s_{i+1} - s_i}{2}\right)^2 \overset{  (a)}{=}\frac{| w_{\max}  -  w_{\min} | ^2 }{ 4 (2^{C_{\text{prec}}} -1)^2},
\end{aligned}
\end{equation}
where $ (a) $ comes from
% \begin{equation}
$|s_{i+1} - s_i|  = \frac{|   w_{\max}  -  w_{\min} | }{2^{C_{\text{prec}} }- 1}$.
% \end{equation}
Thus, we have
\begin{equation}
\begin{aligned}
\expp{\norm{ Q(\bm{w}) - \bm{w}}^2} &= \sum_{i=1}^{d}\expp{\norm{ Q(w_i) - w_i}^2} \leq d \frac{| w_{\max}  -  w_{\min} | ^2 }{ 4 (2^{C_{\text{prec}}} -1)^2}  \leq \frac{d \varepsilon \norm{\bm{w}}^2}{4 (2^{C_{\text{prec}}} -1)^ 2},
\end{aligned}
\end{equation}
where $ \varepsilon = \frac{| w_{\max}  -  w_{\min} |^2}{  \norm{\bm{w}}^2  } $ and $ 0 \leq \varepsilon \leq 1 $. Note that the value of $ \varepsilon $ depends on the skewness of the magnitudes of the entries of $ \bm{w} $. In particular, $ \varepsilon $ increases when $ \bm{w} $ have more skewed entries with a higher variance. $ \varepsilon = 1 $ if and only if $ \bm{w} $ has only one non-zero element, and $ \varepsilon = 0 $ if and only if $ \bm{w}$ have the same magnitude. 

% \section{Additional Notations}\label{app_nota}

%For quantized differential transmission, if $ t+1 \in \mathcal{L}_E$, each device in $ \mathcal{K}_t $ transmit the quantized differential transmission to the CS as $ Q(\bm{g}_t^k) $ where $  \bm{g}_{t}^k = \bm{v}^{k}_{t+1} - \bm{w}_{t+1-E}$. 
\section{Proof of Theorem~\ref{the1}} \label{app3}
In this work, we first introduce some useful notations and then give a useful lemma to prove Theorem~\ref{the1}. 
In our proof, we follow the analysis steps as in \cite{DBLP:journals/jsac/ZhengSC21} and \cite{DBLP:conf/iclr/LiHYWZ20}. Specifically, we denote $ t $ as the round of local iteration and $ \bm{w}_t^k $ as the model parameter of device $ k $ at iteration $ t $. If $ t+1 \in \mathcal{L}$, where $  \mathcal{L} = \{ iI| i=1,2,\cdots\} $, each device transmits their local model updates to the CS and the CS performs global aggregation. If $ t+1 \notin \mathcal{L}$, no global aggregation happens. To make the proof clear, we define the local model update and the global model update as
\begin{equation}
\begin{aligned}
&\bm{v}_{t+1}^k = \bm{w}_t^k - \lambda_t \nabla f_k(\bm{w}_t^{Q, k}, \xi_t^k) ; \\ 
&\bm{u}_{t+1}^k = \left \{
\begin{array}{ll}
\bm{v}_{t+1}^k,              & \text{if }  t+1  \notin  \mathcal{L}_E,\\
\frac{1}{K} \sum_{ i \in \mathcal{K}_{t+1}} \bm{v}_{t+1}^i,     & \text{if }  t+1  \in  \mathcal{L}_E;  \\
\end{array}
\right. \\
&\bm{w}_{t+1}^k  =  \left \{
\begin{array}{ll}
\bm{v}_{t+1}^k,              & \text{if }  t+1  \notin  \mathcal{L}_E,\\
\bm{w}_{t+1-E} +  \frac{1}{K} \sum_{ i \in \mathcal{K}_{t+1}} Q( \bm{d}_{t+1}^i),     & \text{if }  t+1  \in  \mathcal{L}_E;  \\
\end{array}
\right. \\
\end{aligned}
\end{equation}
where $ \bm{d}_{t+1}^k = \bm{v}_{t+1}^k -  \bm{w}_{t+1-E}$, $  \bm{w}_{t+1-E}  $ denotes the most recent global model it downloaded from the CS and $\bm{u}_{t+1}^k  $ is an auxiliary variable. We define three virtual sequences $ \bar{\bm{v}}_t  = \frac{1}{K} \sum_{k=1}^{K} \bm{v}_t^k$, $  \bar{\bm{w}}_t  = \frac{1}{K} \sum_{k=1}^{K} \bm{w}_t^k $ and $  \bar{\bm{u}}_t  = \frac{1}{K} \sum_{k=1}^{K} \bm{u}_t^k $ to facilitate the analysis. Furthermore, we define $ \bar{\bm{g}}_t =  \frac{1}{K} \sum_{k=1}^{K} \nabla f_k(\bm{w}_t^{Q, k}) $ and $ \bm{g}_t  =  \frac{1}{K} \sum_{k=1}^{K} \nabla f_k(\bm{w}_t^{Q, k}, \xi_t^k)$. Thus, we have $ \bar{\bm{v}}_t = \bar{\bm{w}}_t - \lambda_t \bm{g}_t$ and $ \mathbb{E}[\bm{g}_t] = \bar{\bm{g}}_t $. Based on Lemma~\ref{lemma1}, we give an essential lemma which can be considered as an adaptation of \cite{DBLP:journals/jsac/ZhengSC21}:
\begin{lemma}\label{lemma2}
Under Assumption 1, we have
\begin{equation}
\begin{aligned}
\expp{\norm{\bar{\bm{v}}_{t+1} - \bm{w}^*}^2} &\leq (1 - \mu \lambda_t) \expp{\norm{\bar{\bm{w}}_{t} - \bm{w}^*}^2}  + \lambda_t^2\sum_{k=1}^{K} \frac{\sigma_k^2}{K^2} + \frac{d \varepsilon M	^2 (\lambda_t^2 -\lambda_t \mu)}{4 (2^{C_{\text{prec}}} -1)^ 2} + 4 \lambda_t^2 (I-1)^2G^2.
\end{aligned}
\end{equation}
\end{lemma}
Afterward, based on Lemma~\ref{lemma2}, we give the proof of Theorem~\ref{the1} as follows.
\begin{equation}
\begin{aligned}
\norm{ \bar{\bm{w}}_{t+1} - \bm{w}^*}^2 &= \norm{ \bar{\bm{w}}_{t+1} - \bar{\bm{u}}_{t+1} + \bar{\bm{u}}_{t+1}  - \bm{w}^*}^2\\
&=\underbrace{  \norm{  \bar{\bm{w}}_{t+1} - \bar{\bm{u}}_{t+1}  }^2  }_{E_1}+ \underbrace{ \norm{\bar{\bm{u}}_{t+1}  - \bm{w}^* }^2 }_{E_2} + \underbrace{2 \langle  \bar{\bm{w}}_{t+1} - \bar{\bm{u}}_{t+1}  ,\bar{\bm{u}}_{t+1}  - \bm{w}^* \rangle}_{E_3},
\end{aligned}
\end{equation}
For $ E_1 $, we can obtain that
\begin{equation}
\begin{aligned}
\expp{ \bar{\bm{w}}_{t+1}  } &=   \bm{w}_{t+1-E} +  \frac{1}{K} \sum_{ i \in \mathcal{K}_{t+1}} \mathbb{E} [Q( \bm{d}_{t+1}^i) ] = \bm{w}_{t+1-E} +  \frac{1}{K} \sum_{ i \in \mathcal{K}_{t+1}}  \bm{d}_{t+1}^i\\
&=\bm{w}_{t+1-E} +  \frac{1}{K} \sum_{ i \in \mathcal{K}_{t+1}}  ( \bm{v}_{t+1}^i - \bm{w}_{t+1-E}) =  \frac{1}{K} \sum_{ i \in \mathcal{K}_{t+1}}   \bm{v}_{t+1}^i  = \bar{\bm{u}}_{t+1}.
\end{aligned}
\end{equation}
%\\&\leq \frac{d \varepsilon I }{\bar{K}^2 4(2^{C_\text{prec}} - 1) ^2}   \sum_{k\in\mathcal{K}_{t+1}} \sum_{\tau=t+1-I}^{t} \lambda_\tau ^2 \mathbb{E}  \norm{  \nabla F_k( \bm{w} _{\tau}^k, \xi_\tau^k ) }^2
As for the expected norm, we have 
\begin{equation}
\begin{aligned}
\expp{E_1} &= \expp{ \norm{  \bar{\bm{w}}_{t+1} - \bar{\bm{u}}_{t+1}  }^2} \leq \frac{1}{\bar{K}^2} \sum_{k\in\mathcal{K}_{t+1}} \expp{\norm{ Q(\bm{d}^k_{t+1} )- \bm{d}_{t+1}^k }^2} 
 \leq \frac{1}{\bar{K}^2}  \sum_{k\in\mathcal{K}_{t+1}} \frac{d \varepsilon \norm{\bm{d}_{t+1}^k}^2 }{4(2^{C_\text{prec}} - 1)^2  }  \\&\leq \frac{d \varepsilon }{\bar{K}^2 4(2^{C_\text{prec}} - 1) ^2}   \sum_{k\in\mathcal{K}_{t+1}} \norm{ \sum_{\tau=t+1-I}^{t} \lambda_\tau \nabla F_k( \bm{w} _{\tau}^k, \xi_\tau^k ) }^2 \leq \frac{4d \varepsilon I ^2 \lambda_t^2 G^2}{\bar{K}  4(2^{C_\text{prec}} - 1) ^2  } .
\end{aligned}
\end{equation}
Also, $ E_3 $ = 0 after taking expectation. For $ E_2 $, we derive an upper bound by
\begin{equation}
\begin{aligned}
\norm{\bar{\bm{u}}_{t+1}  - \bm{w}^* }^2 &= \norm{\bar{\bm{u}}_{t+1} - \bar{\bm{v}}_{t+1}  +  \bar{\bm{v}}_{t+1} - \bm{w}^*   }^2\\
&= \norm{   \bar{\bm{u}}_{t+1} - \bar{\bm{v}}_{t+1}  }^2 + \norm{ \bar{\bm{v}}_{t+1} - \bm{w}^*  }^2 + \underbrace{2 \langle \bar{\bm{u}}_{t+1} - \bar{\bm{v}}_{t+1}, \bar{\bm{v}}_{t+1} - \bm{w}^*  \rangle}_{E_1}.
\end{aligned}
\end{equation}
Since $ \expp{  \bar{\bm{u}}_{t+1} } = \bar{\bm{v}}_{t+1} $ and $ \expp{\norm{ \bar{\bm{u}}_{t+1} - \bar{\bm{v}}_{t+1}  }  ^2 }   \leq   \frac{ 4(K-\bar{K}) }{ \bar{K}(K-1) } \lambda_t^2I^2G^2 $, we have
\begin{equation}
\begin{aligned}
\expp{E_2} = \expp{\norm{\bar{\bm{u}}_{t+1}  - \bm{w}^* }^2}  \leq   \frac{ 4(K-\bar{K}) }{ \bar{K}(K-1) } I^2G^2\lambda_t^2,
\end{aligned}
\end{equation}
where $E_1  $ becomes zero after taking the expectation. Therefore, we have
\begin{equation}
\begin{aligned}
&\expp{ \norm{ \bar{\bm{w}}_{t+1} - \bm{w}^*}^2  } \leq \expp{ \norm{ \bar{\bm{v}}_{t+1} - \bm{w}^*}^2} +  \frac{4d \varepsilon I ^2 \lambda_t^2 G^2}{\bar{K}  4(2^{C_\text{prec}} - 1) ^2  }    +  \frac{ 4(K-\bar{K}) }{ \bar{K}(K-1) } I^2G^2\lambda_t^2\\
&=(1- \mu\lambda_t) \expp{ \norm{  \bar{\bm{w}}_t -\bm{w}^*}^2 } + \lambda_t^2\sum_{k=1}^{K} \frac{\sigma_k^2}{K^2} + \frac{d}{2^{2n}} (\lambda_t^2-\lambda_t \mu)+ 4 \lambda_t^2 (I-1)^2G^2\\
&\quad + \frac{4d \varepsilon I ^2 \lambda_t^2 G^2}{\bar{K}  4(2^{C_\text{prec}} - 1) ^2  }    +  \frac{ 4(K-\bar{K}) }{ \bar{K}(K-1) }I^2G^2\lambda_t^2.
\end{aligned}
\end{equation}
Since $ \lambda_t > \lambda_t^2 $, we have
\begin{equation}\label{lastd}
\begin{aligned}
&\expp{ \norm{ \bar{\bm{w}}_{t+1} - \bm{w}^*}^2  } \leq \expp{ \norm{ \bar{\bm{v}}_{t+1} - \bm{w}^*}^2} +  \frac{4d I ^2 \lambda_t^2 G^2}{K 2^{2n}}   +  \frac{ 4(K-\bar{K}) }{ \bar{K}(K-1) } I^2G^2\lambda_t^2\\
&=(1- \mu\lambda_t) \expp{ \norm{  \bar{\bm{w}}_t -\bm{w}^*}^2 } + \lambda_t^2\sum_{k=1}^{K} \frac{\sigma_k^2}{K^2}\\
&+  \frac{d \varepsilon M	^2}{4 (2^{C_{\text{prec}}} -1)^ 2} (\lambda_t^2-\lambda_t \mu)+ 4 \lambda_t^2 (I-1)^2G^2 +\frac{4d \varepsilon I ^2 \lambda_t^2 G^2}{\bar{K}  4(2^{C_\text{prec}} - 1) ^2  }    +  \frac{ 4(K-\bar{K}) }{ \bar{K}(K-1). }I^2G^2\lambda_t^2.
\end{aligned}
\end{equation}
Thus, we rewrite \eqref{lastd} as
\begin{equation}
\begin{aligned}
\expp{ \norm{ \bar{\bm{w}}_{t+1} - \bm{w}^*}^2  } \leq (1- \mu\lambda_t) \expp{ \norm{  \bar{\bm{w}}_t -\bm{w}^*}^2 } + \lambda_t^2 D,
\end{aligned}
\end{equation}
\begin{equation}
\begin{aligned}
D &= \sum_{k=1}^{K} \frac{\sigma_k^2}{K^2} + \frac{d \varepsilon M	^2}{4 (2^{C_{\text{prec}}} -1)^ 2}  (1- \mu)+ 4  (I-1)^2G^2+\frac{4d \varepsilon I ^2 G^2}{\bar{K}  4(2^{C_\text{prec}} - 1) ^2  }   +  \frac{ 4(K-\bar{K}) }{ \bar{K}(K-1) }I^2G^2.
\end{aligned}
\end{equation}
%We decay the learning rate with $ \lambda_t = \frac{\beta}{t+\gamma}, \beta > \frac{1}{\mu}  $ and $ \gamma > 0 $ such that $ \lambda_t < 2\lambda_{t+I} $. Then, 
 % $\nu = \max \left\{ \frac{\beta^2 D }{\beta \mu -1}, (\gamma+1)   \expp{ \norm{  \bar{\bm{w}}_t-\bm{w}^*}^2}  \right\}  $ \frac{\nu}{\gamma + t} 
Due to $\expp{ \norm{  \bar{\bm{w}}_t-\bm{w}^*}^2} \leq \frac{\beta^2 D }{(\beta \mu -1)(\gamma+t)} $ satisfies for given learning rate $\lambda_t = \frac{\beta}{t+\gamma}$ \cite{DBLP:conf/iclr/LiHYWZ20}, we obtain Theorem~\ref{the1} by the strong convexity of $ F(\cdot)$,
% \begin{equation}
% \begin{aligned}
$\expp{ F(\bar{\bm{w}}_t)} - F(\bm{w}^*) \leq \frac{L}{2} \expp{ \norm{  \bar{\bm{w}}_t-\bm{w}^*}^2} \leq  \frac{L}{2}\frac{\nu}{\gamma+t} \leq   \frac{L}{2(\gamma+t)}  \frac{\beta^2 D }{\beta \mu -1}.$ Finally, we change the time scale to local iteration.
% \end{aligned}
% \end{equation}

\end{appendices}

\end{document}